\documentclass[final,x11names,hidelinks]{lmcs}
\pdfoutput=1

\usepackage{lastpage}
\lmcsdoi{17}{1}{22}
\lmcsheading{}{\pageref{LastPage}}{}{}%
{Dec.~17,~2019}{Mar.~18,~2021}{}

\usepackage[utf8]{inputenc}


\usepackage{microtype}

\usepackage{bbm}
\usepackage{amssymb}
\usepackage{stmaryrd}
\usepackage{mathtools}
\usepackage{standalone}
\usepackage{subcaption}
\usepackage{multicol}
\usepackage{multirow}
\usepackage{booktabs}
\usepackage{xspace}
\usepackage{tikz}
\usetikzlibrary{arrows, automata, positioning}

\title{Synthesis of Data Word Transducers}
\titlecomment{This paper is the journal version
  of~\cite{DBLP:conf/concur/ExibardFR19}.
  ACM Classification:
Theory of computation $\rightarrow$ Logic and verification;
Theory of computation $\rightarrow$ Automata over infinite objects;
Theory of computation $\rightarrow$ Transducers.}

\author[L. Exibard]{L\'eo Exibard\rsuper{{a,b}}}
\author[E. Filiot]{Emmanuel Filiot\rsuper{b}}
\author[P.-A. Reynier]{Pierre-Alain Reynier\rsuper{a}}

\thanks{
  L.~Exibard is funded by a FRIA fellowship from the F.R.S.-FNRS\@.
  E.~Filiot is a research associate of F.R.S.-FNRS\@. He is supported by the
 ARC Project Transform F\'ed\'eration Wallonie-Bruxelles and the FNRS CDR
 J013116F and MIS F451019F projects.
  P.-A.~Reynier is partly funded by the DeLTA project (ANR--16--CE40--0007).
}

\address{\lsuper{b}Universit\'e libre de Bruxelles, Brussels, Belgium}
\email{\{leo.exibard,efiliot\}@ulb.ac.be}
\address{\lsuper{b}Aix Marseille Univ, Universit\'e de Toulon, CNRS, LIS, Marseille, France}
\email{pierre-alain.reynier@lis-lab.fr}

\newcommand{\N}{\mathbb{N}}
\newcommand{\inp}{\mathbbm{i}}
\newcommand{\outp}{\mathbbm{o}}
\newcommand{\dom}{\mathrm{dom}}
\newcommand{\data}{\mathcal{D}}
\newcommand{\lab}{\textsf{lab}}
\newcommand{\dt}{\textsf{dt}}
\newcommand{\states}{\textsf{states}}
\newcommand{\datawords}{\textsf{DW}}
\newcommand{\findatawords}{\textsf{DW}_{\! f}}
\newcommand{\relwords}{\textsf{RW}}

\newcommand{\real}{\textsc{Real}}
\newcommand{\unif}{\textsc{Unif}}
\newcommand{\allDiff}{\textsc{AllDiff}}
\newcommand{\projin}{\textsf{inp}}
\newcommand{\projout}{\textsf{out}}

\newcommand{\interl}[2]{\ensuremath{\langle #1, #2 \rangle}}

\newcommand{\dz}{\textsf{d}_0}

\newcommand{\dra}{\textnormal{\textsf{DRA}}\xspace}
\newcommand{\draa}{\textnormal{$\textsf{DRA}_{\textsf{ido}}$}\xspace}
\newcommand{\nra}{\textnormal{\textsf{NRA}}\xspace}
\newcommand{\ura}{\textnormal{\textsf{URA}}\xspace}
\newcommand{\nratf}{\textnormal{$\textsf{NRA}_{\textsf{tf}}$}\xspace}
\newcommand{\tf}{\textnormal{$\textsf{tf}$}\xspace}

\newcommand{\srt}{\textnormal{\textsf{RT}}\xspace}
\newcommand{\rt}{\textnormal{\textsf{RT}}\xspace}
\newcommand{\ra}{\textnormal{\textsf{RA}}\xspace}

\newcommand{\tst}{\phi}
\newcommand{\tste}{E}
\newcommand{\asgn}{\textnormal{\textsf{asgn}}}
\newcommand{\Tst}{\textnormal{\textsf{Tst}}}
\newcommand{\Asgn}{\textnormal{\textsf{Asgn}}}
\newcommand{\In}{\textnormal{\textsf{In}}}
\newcommand{\Out}{\textnormal{\textsf{Out}}}
\newcommand{\Comp}{\textnormal{\textsf{Comp}}}
\newcommand{\orig}{o}
\newcommand{\tmm}{\mathcal{M}\xspace}
\newcommand{\req}{\textnormal{\textsf{req}}}
\newcommand{\grt}{\textnormal{\textsf{grt}}}
\newcommand{\idle}{\textnormal{\textsf{idle}}}

\DeclarePairedDelimiter\size{\lvert}{\rvert}
\DeclarePairedDelimiter\eqClass{[}{]}  \newcommand{\inalph}{\Sigma_\inp}
\newcommand{\outalph}{\Sigma_\outp}

\newcommand{\id}{\mathrm{id}} 

\usepackage{calc} \newcommand\myxrightarrow[2][]{
  \xrightarrow[{\raisebox{1.25ex-\heightof{$\scriptstyle#1$}}[0pt]{$\scriptstyle#1$}}]{#2}%
}

\begin{document}

\begin{abstract}
  In reactive synthesis, the goal is to automatically generate an implementation
  from a specification of the reactive and non-terminating input/output
  behaviours of a system. Specifications are usually modelled as logical
  formulae or automata over infinite sequences of signals ($\omega$-words),
  while implementations are represented as transducers. In the classical
  setting, the set of signals is assumed to be finite. In this paper, we
  consider data $\omega$-words instead, i.e., words over an infinite alphabet.
  In this context, we study specifications and implementations respectively
  given as automata and transducers extended with a finite set of registers. We
  consider different instances, depending on whether the specification is
  nondeterministic, universal or deterministic, and depending on whether the
  number of registers of the implementation is given or not.

  In the unbounded setting, we show undecidability for both universal and
  nondeterministic specifications, while decidability is recovered in the
  deterministic case. In the bounded setting, undecidability still holds for
  nondeterministic specifications, but can be recovered by disallowing tests
  over input data. The generic technique we use to show the latter result allows
  us to reprove some known result, namely decidability of bounded synthesis for
  universal specifications.
\end{abstract}

\maketitle

\section*{Introduction}

\emph{Reactive synthesis} is an active research domain whose goal is to design
algorithmic methods able to automatically construct a reactive system from a
specification of its admissible behaviours. Such systems are notoriously
difficult to design correctly, and the main appealing idea of synthesis is to
automatically generate systems that are \emph{correct by construction}. Reactive systems
are non-terminating systems that continuously interact with the environment in
which they are executed, through input and output \emph{signals}. At each time
step, the system receives an input signal from a set $\textsf{In}$ and produces
an output signal from a set $\textsf{Out}$. An execution is then modelled as an
infinite sequence alternating between input and output signals, i.e., an
$\omega$-word in ${(\textsf{In} \cdot \textsf{Out})}^\omega$. Classically, the sets
$\textsf{In}$ and $\textsf{Out}$ are \emph{assumed to be finite} and reactive
systems are modelled as (sequential) \emph{transducers}. Transducers are simple
finite-state machines with transitions of type $\textsf{States}\times
\textsf{In}\rightarrow \textsf{States}\times \textsf{Out}$, which, at any state,
can process any input signal and deterministically produce some output signal,
while possibly moving, again deterministically, to a new state. A
\emph{specification} is then a language $S\subseteq
{(\textsf{In} \cdot \textsf{Out})}^\omega$ telling which are the acceptable behaviours
of the system. It is also classically represented as an automaton, or as a
logical formula then converted into an automaton. Some regular specifications
may not be realisable by any transducer, and the \emph{realisability problem}
asks, given a regular specification $S$, whether there exists a transducer $T$
whose behaviours satisfy $S$ (i.e., are included in $S$). The synthesis problem asks
to construct $T$ if it exists.

A typical example of reactive system is that of a server granting requests from
a \emph{finite} set of clients $C$. Requests are represented as the set of input
signals $\textsf{In} = \{ (r,i)\mid i\in C\} \cup \{ \textsf{idle}\} $ (client
$i$ requests the resource) and grants by the set of output signals
$\textsf{Out} = \{ (g,i)\mid i\in C\} \cup \{ \textsf{idle}\} $ (server grants
client $i$'s request). A typical constraint to be imposed on such a system is
that every request is eventually granted, which can be represented by the LTL
formula $\bigwedge_{i\in C}G((r,i)\rightarrow F (g,i))$. The latter
specification is realisable for instance by the transducer which outputs $(g,i)$
whenever it reads $(r,i)$ and $\textsf{idle}$ whenever it reads $\textsf{idle}$.

It is well-known that the realisability problem is decidable for
$\omega$-regular specifications. It is \textsc{ExpTime}-complete when
represented by parity
automata~\cite{BuLa69,PnuRos:89,DBLP:conf/icalp/FiliotJLW16}; and
\textsc{2ExpTime}-complete for LTL specifications~\cite{PnuRos:89}. Such
positive results have triggered a recent and very active research interest in
efficient symbolic methods and tools for reactive synthesis (see
e.g.~\cite{Bloem2018}). Extensions of this classical setting have been proposed
to capture more realistic scenarios~\cite{Bloem2018}. However, only a few works
have considered infinite sets of input and output signals. In the previous
example, the number of clients is assumed to be finite, and small. To the best
of our knowledge, existing synthesis tools do not handle large alphabets,
although it is more realistic to consider an unbounded (infinite) set of client
identifiers, e.g.\ $C = \N$. The goal of this paper is to investigate how
reactive synthesis can be extended to handle infinite sets of signals.

\emph{Data words} are infinite sequences $x_1x_2\dots$ of labelled data, i.e.,
pairs $(\sigma,d)$ with $\sigma$ a label from a finite alphabet and $d$ is a \emph{data} from a countably infinite alphabet $\data$.
They can naturally model executions of reactive systems over an infinite set of
signals. Among other models, \emph{register automata} are one of the main
extensions of automata recognising languages of data
words~\cite{Kaminski:1994:FA:194527.194534,DBLP:conf/csl/Segoufin06}. They can
use a finite set of registers in which to store data that are read, and to
compare the current data with the content of some of the registers (in this
paper, we allow comparison of equality). Likewise, transducers can be extended
to \emph{register transducers} as a model of reactive systems over data words: a
register transducer is equipped with a set of registers, and when reading an
input labelled data $(\sigma,d)$, it can test $d$ for equality with the content
of some of its registers, and depending on the result of this test,
deterministically assign some of its registers to $d$ and output a finite label
$\beta$ along with the content of one of its registers. Its executions are
then data words alternating between input and output labelled data, and register
automata can thus be used to represent specifications, as languages of such data
words.

\subsubsection*{Contributions} We consider two classical semantics for register
automata, nondeterministic and universal, both with a parity acceptance
condition, which give two classes of register automata respectively denoted \nra
and \ura. We study the parity acceptance condition because it can express the
other classical acceptance conditions; e.g., B\"uchi and co-B\"uchi can be expressed
with a 2-colours parity condition. Since \nra are not closed under complement
(already over finite data words), \nra and \ura define incomparable classes of
specifications. The request-grant specification, as defined above, can be
generalised to an infinite number of clients, and it is then expressible by an
\ura~\cite{DBLP:conf/atva/KhalimovMB18}: whenever a request is made by client
$i$ (labelled data $(r,i)$), universally trigger a run which stores $i$ in some
register and verifies that the labelled data $(g,i)$ eventually occurs in the
data word. In contrast, no \nra can define it. On the other hand, consider the
specification $S_0$: ``all input data but one are copied on the output, the
missing one being replaced by some data which occurred before it'', modelled as
the set of data sequences $d_1d_1d_2d_2\dots d_i d_j d_{i+1} d_{i+1}\dots$ for
all $i\geq 0$ and $j< i$ (finite labels are irrelevant and not represented).
$S_0$ is not definable by any \ura, as it would require to guess $j$, which can
be arbitrarily smaller than $i$, but it is expressible by some \nra making this
guess.

However, we show (unsurprisingly) that the realisability problem by register
transducers of specifications defined by \nra is undecidable. The same negative
result also holds for \ura, solving an open question raised
in~\cite{DBLP:conf/atva/KhalimovMB18}. On the positive side, we show that
decidability is recovered for deterministic (parity) register automata (\dra)
in which the output is driven by the input (meaning that it is contained
in some register). We call this class the \dra with input-driven outputs,
denoted by \draa.
One of the difficulties of register transducer synthesis is that the number of
registers needed to realise the specification is, a priori, unbounded with
regards to the number of registers of the specification. We show it is in fact
not the case for \draa: any specification expressed as a \draa with $r$ registers
is realisable by a register transducer iff it is realisable by a transducer with
$r$ registers.

A way to obtain decidability is to fix a bound $k$ and to target register
transducers with at most $k$ registers. This setting is called \emph{bounded
  synthesis} in~\cite{DBLP:conf/atva/KhalimovMB18}, which establishes that
bounded synthesis is decidable in \textsc{2ExpTime} for \ura. We show that
unfortunately, bounded synthesis is still undecidable for \nra specifications (even when targetting implementations with a single register).
To recover decidability for \nra, we disallow equality tests on the input data
and add a syntactic requirement which entails that on any accepted word, each
output data is the content of some register which has been assigned an input
data occurring before. This defines a subclass of \nra that we call (input)
\emph{test-free} \nra (\nratf). \nratf can express how output data can be
obtained from input data (by copying, moving or duplicating them), although they
do not have the whole power of register automata on the input nor the output
side. Note that the specification $S_0$ given before is \nratf-definable. To
show that bounded synthesis is decidable for \nratf, we establish a generic
transfer property characterising realisable data word specifications in terms of
realisability of corresponding specifications over a finite alphabet, thus
reducing to the well-known synthesis problem over a finite alphabet. Such
property also allows us to reprove the result
of~\cite{DBLP:conf/atva/KhalimovMB18}, with a rather short proof based on
standard results from the theory of register automata, indicating that it might
allow to establish decidability for other classes of data specifications. Our
results are summarised in Table~\ref{tbl:results}.

\begin{table}
  \makebox[\textwidth][c]{%
    \begin{tabular}{ccccc}
                                   & \draa                          & \nra                          & \ura                                                               & \nratf \\[2mm]
\toprule
Bounded & \textsc{2ExpTime}         & Undecidable\ ($k\geq 1$)      & \textsc{2ExpTime}                                                  & \textsc{2ExpTime} \\
           Synthesis                        & (Thm.~\ref{thm:boundedURA})   & (Thm.~\ref{thm:boundedNRA})   & (\cite{DBLP:conf/atva/KhalimovMB18} and Thm.~\ref{thm:boundedURA}) & (Thm.~\ref{thm:boundedNRAtf}) \\[2mm]
     General      & \textsc{ExpTime-c}              & Undecidable                   & Undecidable                                                        & \multirow{2}{*}{Open} \\
        Case            & (Thm.~\ref{thm:unboundedDRA}) & (Thm.~\ref{thm:unboundedNRA}) & (Thm.~\ref{thm:unboundedURA})                                      &
    \end{tabular}%
    }
  \caption{Decidability status of the problems studied. As observed in Corollary~\ref{cor:boundedDRA}, the bounded synthesis for \draa is in \textsc{ExpTime}
if the target number of registers is larger than or equal to the number of registers of the specification.}\label{tbl:results}
\end{table}

\subsubsection*{Related Work} As already mentioned, bounded synthesis of register
transducers is considered in~\cite{DBLP:conf/atva/KhalimovMB18} where it is
shown to be decidable for \ura. We reprove this result in a shorter way. Our
proof bears some similarities with that of~\cite{DBLP:conf/atva/KhalimovMB18},
but it seems that our formulation benefits more from the use of existing
results. The technique is also more generic and we instantiate it to \nratf.
\nratf correspond to the one-way, nondeterministic version of the expressive
transducer model of~\cite{DBLP:conf/fossacs/Durand-Gasselin16}, which however
does not consider the synthesis problem.

The synthesis problem over infinite alphabets is also
considered in~\cite{Ehlers:2014:SI:2961203.2961226}, in which data
represent identifiers and specifications (given as particular
automata close to register automata) can depend on equality between
identifiers. However, the class of implementations is very expressive:
it allows for unbounded memory
through a queue data structure. The synthesis problem is shown to be
undecidable and a sound but incomplete algorithm is given.

Finally, classical reactive synthesis has strong connections with game theory
on finite graphs. Some extension of games to infinite graphs whose
vertices are valuations of variables in an infinite data domain have been
considered in~\cite{DBLP:conf/lics/FigueiraP18}. Such games are shown
to be undecidable and a decidable restriction is proposed, which
however does not seem to match our context.

\section{Data Words and Register Automata}

For a (possibly infinite) set $S$, we denote by $S^\omega$ the set
of infinite  words over this alphabet.
For $1\leq i\leq j$, we let $u[i{:}j] = u_i u_{i+1} \dots u_j$
and $u[i] = u[i{:}i]$ the $i$th
letter of $u$.
For $u,v \in S^\omega$, we define their
\emph{interleaving} $\interl{u}{v} = u[1]v[1]u[2]v[2]\dots$

\subsection{Data Words}
Let $\Sigma$ be a finite alphabet and
$\mathcal{D}$ a
countably
infinite set, denoting, all over this paper,
a set of
elements called \emph{data}. We also distinguish an (arbitrary) data value
$\dz\in\data$. Given a set $R$, let
$\tau_0^R$ be the constant function defined by $\tau_0^R(r)=\dz$ for all
$r\in R$. A \emph{labelled data} (or l-data for
short) is a pair $x = (\sigma,d)\in\Sigma\times \data$, where $\sigma$
is the \emph{label} and $d$ the \emph{data}. We
define the projections $\lab(x) = \sigma$ and $\dt(x) = d$.
A \emph{data word} over $\Sigma$ and
$\mathcal{D}$ is an infinite sequence of labelled data, i.e.\ a word
$w\in {(\Sigma\times \mathcal{D})}^\omega$. We extend the projections
$\lab$ and $\dt$ to data words naturally,
i.e.\ $\lab(w)\in\Sigma^\omega$ and $\dt(w) \in \data^\omega$. We
denote the set of data words over $\Sigma$ and $\data$ by $\datawords(\Sigma,\data)$
($\datawords$ when
clear from the context). A \emph{data word language} is a subset
$L\subseteq \datawords(\Sigma,\data)$. Note that in this paper, data
words are infinite, otherwise they are called \emph{finite data
  words}, and we denote by $\findatawords(\Sigma,\data)$ the set of
finite data words.

\subsection{Register Automata}
Register automata are automata
recognising data word languages. They were first
introduced in~\cite{Kaminski:1994:FA:194527.194534} as finite-memory
automata. Here, we define them in a spirit close
to~\cite{journals/jcss/LibkinTV15}, but over infinite words,
with a parity acceptance condition. The current data can be
compared for equality with the register contents via tests. Our tests
are symbolic and defined via Boolean formulas of the following
form. Given $R$ a set of registers, a \emph{test} is a formula $\phi$
satisfying the following syntax:
\[
\phi\ ::=\ \top\mid \bot\mid r^=\mid r^{\neq} \mid \phi\wedge \phi \mid \phi\vee
\phi\mid \neg \phi
\]
where $r\in R$. Given a valuation $\tau : R\rightarrow \data$, a
test $\phi$ and a data $d$, we denote by $\tau,d\models \phi$ the
satisfiability of $\phi$ by $d$ in valuation $\tau$, defined as $\tau,d\models r^=$
if $\tau(r) = d$ and $\tau,d\models r^{\neq}$ if $\tau(r)\neq d$. The
Boolean combinators behave as usual. We denote by $\Tst_R$ the set of
(symbolic) tests over $R$.

\begin{defi}
  A \emph{register automaton} (\ra) is a tuple $\mathcal{A} = (\Sigma,
  \mathcal{D}, Q, q_0, \delta, R, c)$, where:
  \begin{itemize}
  \item $\Sigma$ is a finite alphabet of \emph{labels}, $\mathcal{D}$ is an infinite alphabet of \emph{data}
  \item $Q$ is a finite set of states and $q_0 \in Q$ is the initial state
  \item $R$ is a finite set of \emph{registers}. We denote $\Asgn_R = 2^R$.

  \item $c : Q \rightarrow \{1, \dots, d\}$, where $d \in \N$ is the number of
    \emph{priorities}, is the \emph{colouring function}, used to define the
    acceptance condition
  \item $\delta \subseteq Q \times \Sigma \times \Tst_R \times \Asgn_R \times Q$ is a set of transitions.
  \end{itemize}
\end{defi}

\noindent
  A transition $(q, \sigma, \tst, \asgn, q')$ is also written $q \myxrightarrow[\mathcal{A}]{\sigma, \tst, \asgn} q'$. We may omit $\mathcal{A}$ in the latter notation. Intuitively such transition means that on input ($\sigma, d$) in state $q$ the automaton:
    \begin{enumerate}
    \item checks that $\tst$ is satisfied by the current register
      contents and the current data
    \item assigns $d$ to all the registers in $\asgn$ ($\asgn$ might be empty)
    \item transitions to state $q'$.
    \end{enumerate}

\noindent
    $\mathcal{A}$ is said to be \emph{deterministic} if the tests are
    mutually exclusive, i.e., for any two distinct transitions of the form $q
    \xrightarrow{\sigma, \tst, \asgn} q'$ and $q \xrightarrow{\sigma',
      \tst', \asgn'} q''$, then either
    $\sigma\neq \sigma'$ or $\tst\wedge \tst'$ is not satisfiable. The
    automaton $\mathcal{A}$ is said to be \emph{complete} if
    for any given state $q$, any label $\sigma$, any data $d$ and any register
    valuation $\tau$, there exists a transition $q\xrightarrow{\sigma, \tst,\asgn} q'\in
    \delta$ such that $\tau,d\models \tst$.


\subsection{Configurations and Runs}\label{def:CfgsRuns}
A configuration is a pair $(q, \tau) \in Q \times (R \rightarrow
\mathcal{D})$.  
Fix a transition $t = p \xrightarrow{\sigma, \tst, \asgn} p'$. We say that
$(q,\tau)$ enables $t$ on reading $(\sigma',d)$ if $q = p$, $\sigma' = \sigma$
and $\tau,d\models \tst$. Let $\text{next}(\tau,\asgn,d)$ be the valuation
$\tau'$ defined by $\tau'(i)=d$ if $i\in\asgn$, and $\tau'(i) = \tau(i)$
otherwise.
We extend this notation to configurations as follows: if $\gamma = (q,\tau)$
enables $t$ on input $(\sigma,d)$, the \emph{successor} configuration of $(q,
\tau)$ by $t$ on input $(\sigma,d)$ is $\text{next}(\gamma,\asgn,d,t) = (p',
\text{next}(\tau,\asgn,d))$. We also write $\text{next}(\gamma,t,\sigma,d)$ to
denote the successor of $(q,\tau)$ by transition $t$ when $(q,\tau)$ enables $t$
on input $(\sigma,d)$. The \emph{initial configuration} is $(q_0, \tau_0^R)$.
Then, a \emph{run} over a data word $(\sigma_1, d_1) (\sigma_2, d_2) \dots$ is
an infinite sequence of transitions $t_0 t_1\dots$ such that there exists a
sequence of configurations $\gamma_0\gamma_1\dots = (q_0, \tau_0) (q_1, \tau_1)
\dots$ such that $\gamma_0$ is initial and for all $i\geq 0$, $\gamma_{i+1} =
\text{next}(\gamma_i,t_i,\sigma_i,d_i)$. With a run $\rho$, we associate its
sequence of states $\states(\rho) = q_0 q_1 \dots$

\subsection{Languages Defined by \ra}
Given a run $\rho$, we denote, by a
slight abuse of notation, $c(\rho) = \max \{j \mid c(q_l) = j \text{ for
  infinitely many } q_l \in \states(\rho)\}$ the maximum color that occurs
infinitely often in $\rho$. Then, in the parity acceptance condition, $\rho$ is
accepting whenever $c(\rho)$ is even. We consider two dual semantics for \ra:
nondeterministic (N) and universal (U). Given a \ra $A$, depending on whether
it is considered nondeterministic or universal, it recognises $L_{N}(A) = \{w
\mid \text{there exists an accepting run $\rho$ on $w$}\}$ or $L_{U}(A) = \{w
\mid \text{all runs $\rho$ on $w$ are accepting}\}$.
Note that those semantics are dual: for a \ra $A$, by letting $\overline{A}$ be
a copy of $A$ with colouring function $\overline{c}: q \mapsto c(q)+1$, we have that $L_{U}(\overline{A}) = \overline{L_{N}(A)}$.

We denote by \nra (resp. \ura) the
class of register automata interpreted with a nondeterministic (resp.
universal) parity acceptance condition, and given $A\in\nra$ (resp. $A\in
\ura$), we write $L(A)$ instead of $L_{N}(A)$ (resp. $L_{U}(A)$). We also
denote by \dra the class of deterministic parity register automata.

\section{Synthesis of Register Transducers}%
\label{sec:synthesisProblem}

\subsection{Specifications, Implementations and the Realisability Problem}%
\label{sec:specImplReal}
Let $\inalph$ and $\outalph$ be two finite alphabets of labels, and $\data$ a
countable set of data. A \emph{relational data word} is an element of $w\in
{[(\inalph\times \data) \cdot (\outalph\times \data)]}^\omega$. Such a word is called
relational as it defines a pair of data words in
$\datawords(\inalph,\data)\times \datawords(\outalph,\data)$ through the
following projections. If $w = x^1_\inp x^1_\outp x^2_\inp x^2_\outp\dots$, we
let $\projin(w) = x^1_\inp x^2_\inp\dots$ and $\projout(w) = x^1_\outp
x^2_\outp\dots$ We denote by $\relwords(\inalph,\outalph,\data)$ (just
$\relwords$ when clear from the context) the set of relational data words. A
\emph{specification} is simply a language $S\subseteq
\relwords(\inalph,\outalph,\data)$.
An \emph{implementation} is a total function $I : {(\inalph\times
\data)}^*\rightarrow \outalph\times \data$. From $I$, we define another function
$f_I : \datawords(\inalph, \data) \rightarrow \datawords(\outalph,\data)$ which,
with an input data word $w_\inp = x^1_\inp x^2_\inp\dots \in \inalph\times
\data$, associates the output data word $f_I(w_\inp) = x^1_\outp x^2_\outp
\dots$ such that $\forall i\geq 1$, $x^i_\outp = I(x^1_{\inp} \dots
x^{i{-}1}_\inp)$. $I$ also defines a language of relational data words $L(I) =
\{\interl{w_\inp}{f_I(w_\inp)}\mid w_\inp\in \datawords(\inalph,\data)\}$.

We say that $I$ \emph{realises} $S$ when $L(I)\subseteq S$, and that $S$ is
\emph{realisable} if there exists an implementation realising it. Note that
since $f_I$ is a total function, we have that if $S$ is realisable, then in
particular its domain is total, i.e.\ for all $w_\inp \in \datawords(\inalph,
\data)$, there exists $w_\outp \in \datawords(\outalph, \data)$ such that
$\interl{w_\inp}{w_\outp} \in S$. Therefore, any specification whose domain is
not total is not realisable according to this definition. For a discussion on
this definition, see Section~\ref{sec:domain}.

The \emph{realisability problem} consists, given a (finite representation of a)
specification $S$, in checking whether $S$ is realisable. In general, we
parameterise this problem by classes of specifications $\mathcal{S}$ and
of implementations $\mathcal{I}$, defining the
$(\mathcal{S},\mathcal{I})$-realisability problem, denoted
$\real(\mathcal{S},\mathcal{I})$. Given a specification $S\in \mathcal{S}$, it
asks whether $S$ is realisable by some implementation $I \in \mathcal{I}$.
We now introduce the classes $\mathcal{S}$ and
$\mathcal{I}$ we consider.

\subsection{Specification Register Automata}
In this paper, we
consider specifications defined by register automata (hence alternately
reading input and output labelled data). We assume that
the set of states is partitioned into $Q_\inp$ (called input states,
reading only labels in $\inalph$) and $Q_\outp$ (called output states,
reading only labels in $\outalph$), where $q_0 \in Q_\inp$,
and such that the transition relation $\delta$
alternates between these two sets,
i.e.\ \[\delta \subseteq \bigcup\nolimits_{\alpha = \inp, \outp}(Q_\alpha \times \Sigma_\alpha \times \Tst_R \times \Asgn_R
\times Q_{\overline{\alpha}}),\] where $\overline{\inp} = \outp$ (resp. $\overline{\outp} = \inp$).
We denote by $\dra$ (resp. $\nra$,
$\ura$) the class of specifications defined by deterministic (resp.
nondeterministic, universal) parity register automata.
\begin{exa}%
  \label{exa:reqGrant}
  Remember the setting described in the introduction of a server granting
  requests from an unbounded set of clients $C$. The input (resp.\ output) finite
  alphabets are $\inalph = \{\req, \idle\}$ and $\outalph = \{\grt, \idle\}$,
  while the set of data is any countably infinite set $\data$ containing $C$.
  Without loss of generality, $C \subseteq \N$ is a set of client ids, so we can
  take $\data = \N$. Then, as stated in the introduction, the specification that
  for all $i \in C$, every request of client $i$ is
  eventually granted can be expressed with the \ura of
  Figure~\ref{fig:reqGrant}.
  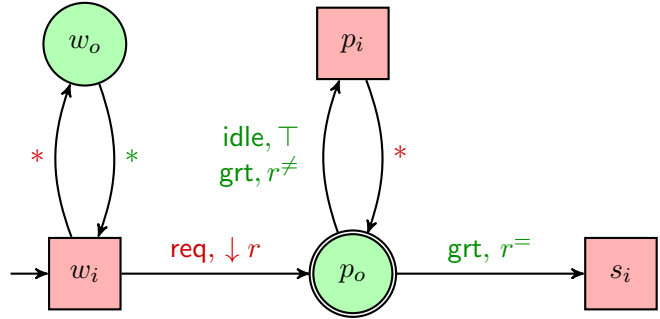
\begin{figure}[ht]
    \centering
    \begin{tikzpicture}[->,>=stealth',auto,node distance=2.5cm,thick]
      \tikzstyle{every state}=[text=black]
      \tikzstyle{input}=[rectangle,fill=red!30,minimum size=9.5mm]
      \tikzstyle{output}=[fill=green!30,minimum size=11mm] \tikzstyle{input
        char}=[text=Red3] \tikzstyle{output char}=[text=Green4] \node[state,
      initial, initial text={}, input] (ii) {$w_i$}; \node[state, above=2cm of
      ii, output] (io) {$w_o$}; \node[state, right= of ii, output, accepting]
      (po) {$p_o$}; \node[state, above=2cm of po, input] (pi) {$p_i$};
      \node[state, right= of po, input] (si) {$s_i$};

      \path (ii) edge node[above,input char] {\req, $\downarrow{} r$} (po);
      \path (po) edge node[above, output char] {\grt, $r^=$} (si);
      \path (ii) edge[bend left=20] node[left, input char] {*} (io);
      \path (po) edge[bend left=20] node[left, output char] {$\begin{array}{r} \idle,\top \\ \grt, r^{\neq} \end{array}$} (pi);
      \path (pi) edge[bend left=20] node[right, input char] {*} (po);
      \path (io) edge[bend left=20] node[right, output char] {*}
      (ii);
    \end{tikzpicture}
    \caption{A universal register automaton checking that every request is
      eventually granted. Input is in red (states are squares), output is in
      green (states are circles). Finite labels are \textsf{sans
        serif}.
      All states have priority $0$, except the
      doubly circled state $p_o$, which has priority $1$. This corresponds to a
      co-B\"uchi acceptance condition with rejecting state $p_o$. The automaton
      always loops between $w_i$ and $w_o$ (the $*$ symbol means that the
      transition is taken, no matter the labelled-data received). Whenever it
      receives a request as input, it universally spawns a run which stores the
      corresponding id in its single register $r$ (depicted as $\downarrow r$),
      and transitions to $p_0$. Then, it loops between $p_i$ and $p_o$ while it
      does not receive the corresponding grant, with matching id, as output
      (i.e.\ either reads $\idle$ or receives a grant with wrong id: $d \neq r$).
      When it receives a grant with the right id ($d = r$), it transitions to $s_i$,
      then the run dies at the next step (which favors acceptance in the
      universal semantics).}%
    \label{fig:reqGrant}
  \end{figure}
\end{exa}

\subsection{Register Transducers As Implementations}
We consider
implementations represented as transducers processing data words. A
\emph{register transducer} is a tuple $T = (\inalph,\outalph, Q, q_0, \delta,
R)$ where $Q$ is a finite set of states with initial state $q_0$, $R$ is a
finite set of registers, and $\delta : Q \times \inalph \times \Tst_R
\rightarrow \Asgn_R \times \outalph \times R \times Q$ is the transition
function (as before, $\Asgn_R = 2^R$), assumed to be complete in the
sense that, as for \ra, for every state $q$ and label $\sigma_\inp$, for every data
$d$ and register valuation $\tau$, there exists a transition
$\delta(q,\sigma_\inp,\phi) = (\asgn,\sigma_\outp,r,q')$ such that
$\tau,d\models \phi$. When processing an l-data
$(\sigma_\inp,d)$, $T$ compares $d$ with the content of some of its registers,
and depending on the result, moves to another state, stores $d$ in some
registers, and outputs some label in $\outalph$ along with the content of some
register $r\in R$.

Let us formally define the semantics of a register transducer $T$, as an
implementation $I_T$. First, for a finite input data word $w =
(\sigma_{\inp}^1,d_\inp^1)\dots (\sigma_\inp^n,d_\inp^n)$ in ${(\inalph \times
\data)}^*$, we denote by $(q_i,\tau_i)$ the $i$th configuration reached by $T$ on
$w$, where $(q_0,\tau_0)$ is initial and for all $0<i< n$, $(q_i,\tau_i)$ is the
unique configuration such that there exists a transition
$\delta(q_{i-1},\sigma_\inp^{i}, \tst) = (\asgn,\sigma_\outp, r, q_i)$ such that
$\tau_{i-1},d_\inp^{i}\models \tst$ and $\tau_i =
\text{next}(\tau_{i-1},d_\inp^i,\asgn)$. We let $(\sigma_\outp^i,d_\outp^i) =
(\sigma_\outp, \tau_i(r))$ and $I_T(w) = (\sigma_\outp^n,d_\outp^n)$. Then, we
denote $f_T = f_{I_T}$ and $L(T) = L(I_T)$. Note that if $T$ is interpreted as a
\dra with exactly one transition per output state and whose states are all
accepting (i.e.\ have even maximal parity $0$), then $L(I_T)$ is indeed the
language of such register automaton. We denote by $\rt[k]$ the class of
implementations defined by register transducers with at most $k$ registers, and
by $\rt = \bigcup_{k\geq 0} \rt[k]$ the class of implementations defined by
register transducers.
\begin{exa}%
  \label{exa:grantReqs}
  Consider again the specification of Example~\ref{exa:reqGrant}. Such
  specification is realisable for instance by the transducer which outputs
  $(\grt,i)$ whenever it reads $(\req,i)$ and $(\idle,d)$ ($d$ does not matter)
  whenever it reads $\idle$,
  which is depicted in Figure~\ref{fig:grantReqs}.
  \begin{figure}[ht]
    \centering
\begin{tikzpicture}[->,>=stealth',auto,node distance=2.5cm,thick]
        \tikzstyle{every state}=[fill=yellow!30,text=black]
\newcommand{\inColored}[1]{\textcolor{Red3}{#1}}
\newcommand{\outColored}[1]{\textcolor{Green4}{#1}}

\node[state, initial, initial text={}] (i) {};

\path (i) edge[loop above] node[above] {$\inColored{\req}, \inColored{\top} \mid
  \inColored{\downarrow r}, \outColored{\grt}, \outColored{\uparrow r}$} (i);
\path (i) edge[loop below] node[below] {$\inColored{\idle}, \inColored{\top}
  \mid \outColored{\idle},\outColored{\uparrow r}$} (i);
\end{tikzpicture}
\caption{A register transducer immediately granting each request. The notations
  are the same as in Figure~\ref{fig:reqGrant}. Additionally, here, $\uparrow
  r$ means that the transducer outputs the content of $r$.}%
\label{fig:grantReqs}
\end{figure}
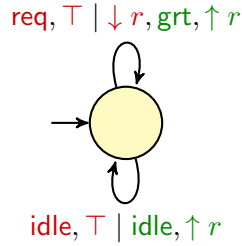
\end{exa}

\subsection{Synthesis from Data-Free Specifications}
If in the latter
definitions of the synthesis problem, one considers specifications defined by \ra with no
registers (i.e.\ parity automata), and implementations defined by \srt with no
registers, then the data in data-words can be ignored and we are back to the
classical reactive synthesis setting, for which important results are known:
\begin{thmC}[\cite{BuLa69}]\label{thm:reactivesynthesis}
  The realisability problem of (data-free) specifications given as (register-free)
  nondeterministic parity automata by (register-free) transducers is {\normalfont\textsc{ExpTime}}-complete.
\end{thmC}
\begin{proof}
  The upper bound was first established in~\cite{BuLa69} and~\cite{PnuRos:89}.
  Hardness is folklore, but a proof in the particular case of finite words
  (easily adapted to the $\omega$-word setting) can be found
  in~\cite[Proposition 6]{DBLP:conf/icalp/FiliotJLW16}.
\end{proof}

\section{Unbounded Synthesis}

In this section, we consider the unbounded synthesis problem $\real(\ra,\rt)$.
Thus, we do not fix a priori the number of registers of the
implementation.

\subsection{Undecidability Results}

Let
us first consider the case of \nra and \ura, which are, in our setting, the most
natural devices to express data word specifications.
Unfortunately, the two corresponding problems happen to be undecidable:
\begin{thm}%
  \label{thm:unboundedNRA}
  {\normalfont$\real(\nra,\rt)$} is undecidable.
\end{thm}
\begin{proof}
  We reduce the problem from the universality of \nra over finite words, which
  is undecidable~\cite{Neven:2004:FSM:1013560.1013562}. Let $A$ be a (finite
  data-word) \nra. Let $S$ be a specification which first reads some finite data
  word $w$, then a separator $\#$ (its associated data is arbitrary and not
  represented), then allows for swapping the first and second l-data on any
  input read later on, while also allowing to behave like the identity whenever
  $w \in L(A)$. $S$ is also equal to the identity over any word not containing
  $\#$ so that its domain is total. Formally, let $S = S_1 \cup S_2 \cup T$,
  where:
  \begin{align*}
  S_1 &=  \left\{(w \# (\sigma_1, d_1) (\sigma_2, d_2) u, w \# (\sigma_2, d_1)
  (\sigma_1, d_2) u) \;\middle|\; \begin{array}{c}d_1,d_2\in \data, \sigma_1,\sigma_2 \in \Sigma \\ w \in \findatawords, u \in \datawords \end{array}\right\}\\
  S_2 &=  \{(w
  \# u, w \# u) \mid w \in L(A), u \in \datawords\}\\ 
  T &=  \{(w,w) \mid w \notin \findatawords \# \datawords\} 
  \end{align*}
  $S$ is definable by a $\nra$ running
  over relational data words, because each component is and $\nra$ are closed
  under union. Recognising the interversion of the first two labels $\sigma_1$
  and $\sigma_2$ after the $\#$ in $S_1$ is easily done using nondeterminism,
  and the behaviour on data is the identity, so $S_1$ is \nra-definable. Then,
  emulating the identity over some \nra-definable domain is easy, so $S_2$ and
  $T$ are also \nra-definable.

  Now, if $A$ is universal, ie $L(A) = \findatawords$, then the identity
  $\id_\datawords$ over $\datawords$ realises $S$, since then $\id_\datawords
  \subseteq S$ and has total domain. Conversely, if $L(A) \subsetneq
  \findatawords$, assume by contradiction that $S$ is realisable by a register
  transducer $I$. Let $w \in \findatawords \backslash L(A)$. Then, for any
  $(\sigma_1, d_1) (\sigma_2, d_2) u \in \datawords$, we must have $I(w \#
  (\sigma_1, d_1) (\sigma_2, d_2) u) = w \# (\sigma_2, d_1) (\sigma_1, d_2) u$;
  but this implies guessing the second label while having only read the first
  one, which is not doable by any transducer as long as $\sigma_1 \neq
  \sigma_2$.
\end{proof}

Actually, we can observe that such undecidability proof extends to
$\real(\nra,\rt[1])$, and to all $\real(\nra,\rt[k])$ for $k \geq 1$.
Indeed, $A$ is universal iff $S$ is realisable by the identity over data words,
which is implementable using a $1$-register transducer:
\begin{thm}%
  \label{thm:boundedNRA}
  For all $k \geq 1$, {\normalfont$\real(\nra,\rt[k])$} is undecidable.
\end{thm}

Now, we can show that the unbounded synthesis problem is also undecidable for \ura, answering a question left open in~\cite{DBLP:conf/atva/KhalimovMB18}.


\begin{thm}%
  \label{thm:unboundedURA}
  {\normalfont$\real(\ura,\rt)$} is undecidable.
\end{thm}

\begin{proof}
  We present a reduction to our synthesis problem from the emptiness problem of
  \ura over finite words. The latter is undecidable
  by a direct reduction
  from the universality problem of \nra, which is undecidable
  by~\cite{Neven:2004:FSM:1013560.1013562}.

  First, consider the relation $S_1 = \{(u \# v, u \# w) \mid u \in
  \findatawords, v \in \datawords$, each data of $u$ appears infinitely
    often in $w\}$. $S_1$ is recognised by a $1$-register \ura which, upon
  reading a data $d$ in $u$, stores it in its register and checks that it
  appears infinitely often in $w$ by visiting a state with maximal parity $2$
  every time it sees $d$ (all other states have parity $1$). Note that for all
  $k \geq 1$, $S_1 \cap \{(u \# v, u \# w) \mid u \in \findatawords, v,w \in
  \datawords \text{ and $u$ has at most $k$ distinct data}\}$ is realisable by a
  $k$-register transducer: on reading $u$, store each distinct data in one
  register, and after the $\#$ output them in turn in a round-robin fashion.
  However, $S_1$ is not realisable: on reading the $\#$ separator, any
  implementation must have all the data of $u$ in its registers, but the number
  of such data is not bounded ($u$ can have pairwise distinct data and be of
  arbitrary length).

 Then, let $A$ be a \ura over finite data words. Consider the specification $S = S_1 \cup S_2 \cup T$, where $S_2 = \{(u \# v, u \# w \# {(a,\dz)}^\omega) \mid u \in \findatawords, v \in \datawords, w \in L(A)\}$ and $T = \{(u,w) \mid u \notin \findatawords \# \datawords, w \in \datawords\}$. $S$ has total domain, and is recognisable by a \ura. Indeed, \ura are closed under union, by the same product construction as for the intersection of \nra~\cite{Kaminski:1994:FA:194527.194534}, and each part is \ura-recognisable: $S_1$ is, as described above, $S_2$ is by simulating $A$ on the output to check $w \in L(A)$ then looping over $(a, \dz)$, and $T$ simply checks a regular property.

 Now, if $L(A) \neq \varnothing$, let $w \in L(A)$ and let $D_w = \{d_1, \dots,
 d_k\}$ be the set of data distinct from $\dz$ that occur in $w$. As a
 consequence of the closure under automorphisms of register automata~\cite[Proposition
 2]{Kaminski:1994:FA:194527.194534}, we have: for any set $D
 \subseteq \data$ such that $\size{D} \geq k$, and for any injection $\pi : D_w
 \cup \{\dz\}
 \rightarrow D \cup \{\dz\}$ such that $\pi(\dz) = \dz$, by
 extending $\pi$ to a morphism $\widehat{\pi}$ over data words in the usual way
 (and behaving as the identity over the finite labels),
 $\widehat{\pi}(w) \in L(A)$. Indeed, as register automata can only test for
 equality, acceptance is determined by the equality relations between the
 different data of the input, so we can rename them (with the exception of
 $\dz$, which is a distinguished data).

 Then, $S$ is realisable by a register transducer $I$ with $k+2$ registers.
 While it has not read a $\#$, $I$ reads its input $u$ and outputs it along the
 way, using one register to store the current data and output it immediately.
 Meanwhile, it also stores the first $k$ distinct data of $u$ in its registers.
 Its last register is used to keep $\dz$ in memory. If there is no $\#$ in the
 input, then $I(u) = u$, so $(u,I(u)) \in T$. Now, if some $\#$ is read, $I$
 outputs $\#$ (along with an arbitrary data), and there are two cases: if the
 number of data in $u$ is lower than or equal to $k$, $I$ realises $S_1$, as
 described above. Otherwise, let $D_u = \{e_1, \dots, e_l\}$ be the set of data
 of $u$ distinct from $\dz$, indexed by order of appearance $(l \geq k)$. Then,
 let $\pi : D_w \cup \{\dz\} \rightarrow D_u \cup \{\dz\}$ be such that for all
 $1 \leq i \leq k, \pi(d_i) = e_i$ and $\pi(\dz) = \dz$: $\pi$ is injective.
 Now, $I$ can output $\widehat{\pi}(w) \# {(a, \dz)}^\omega$ since it stored
 $\{e_1, \dots, e_k\}$ in its registers, hence realising $S_2$. Conversely, if
 $L(A) = \varnothing$, then $S$ is not realisable. If it were, $S \cap
 \findatawords \# \datawords = S_1$ would be too, as a regular domain
 restriction, but we have seen above that this is not the case. Thus, $S$ is
 realisable iff $L(A) \neq \varnothing$.
\end{proof}

\subsection{\texorpdfstring{A Decidable Subclass: \draa}{A Decidable Subclass: DRAido}}

However, we show that restricting to \dra allows to recover decidability,
modulo one additional assumption, namely that the output data of a
transition has to be the content of some register. We formally
define this class as follows:

\begin{defi}[\draa]
Let $\mathcal{A} = (\Sigma, \mathcal{D}, Q, q_0, \delta, R, c)$
be a \dra. We say that $\mathcal{A}$ is with \emph{input-driven outputs}
if for any
output transition $p\xrightarrow{\sigma,\tst,\asgn} q$, the test
$\tst$ is of the form $r^=$ for some $r\in R$. We denote by $\draa$
the class of \dra with input-driven outputs.
\end{defi}

Such
assumption rules out pathological, and to our opinion uninteresting and
technical cases stemming from the asymmetry between the class of specifications
and implementations. E.g., consider the single-register \dra in
Fig.~\ref{fig:cornerCaseDRA1} (finite labels are arbitrary and not depicted). It
starts by reading one input data $d$ and stores it in $r$, asks that the
corresponding output data is different from the content $d$ of $r$, then accepts any output over any input
(transitions $\top$ are always takeable). It is not realisable because
transducers necessarily output the content of some register (hence producing a
data which already appeared). On the other hand, having tests of the
form $\tst =
r^{\neq}$ for instance does not imply unrealisability, as shown by the \dra of
Fig.~\ref{fig:cornerCaseDRA2}: it starts by reading one data $d_1$, asks to copy
it on the output, then reads another data $d_2$, and requires that the output is
either distinct from $d_1$ or equal to it, depending on whether $d_2 \neq d_1$.
It happens that such specification is realisable by the identity.

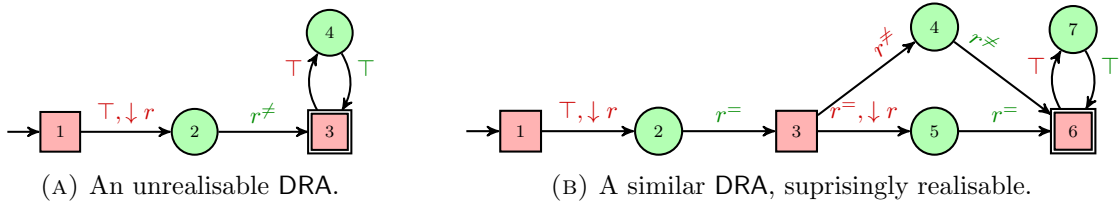
\begin{figure}[ht]
  \centering
  \begin{subfigure}[t]{.35\textwidth}
    \centering
    \resizebox{\textwidth}{!}{%
      \begin{tikzpicture}[->,>=stealth',auto,node distance=1.25cm,thick,scale=0.9,every node/.style={scale=0.85}]
        \tikzstyle{every state}=[text=black, font=\scriptsize]
        \tikzstyle{input}=[rectangle,fill=red!30,minimum size=6.5mm]
        \tikzstyle{output}=[fill=green!30,minimum size=7.5mm]
        \tikzstyle{input char}=[text=Red3]
        \tikzstyle{output char}=[text=Green4]

        \node[state, initial, initial text={}, input] (i) {1};
        \node[state, right = of i, output]            (p) {2};
        \node[state, right = of p, input, accepting]  (q) {3};
        \node[state, above = 0.75cm of q, output]     (r) {4};

        \path (i) edge            node[input char]                   {$\top, \downarrow {} r$} (p);
        \path (p) edge            node[output char]                  {$r^{\neq}$} (q);
        \path (q) edge[bend left] node[left,near end,input char]     {$\top$} (r);
        \path (r) edge[bend left] node[right,near start,output char] {$\top$} (q);
      \end{tikzpicture}
    }
  \caption{An unrealisable \dra.}%
  \label{fig:cornerCaseDRA1}
  \end{subfigure}\hfill %
  \begin{subfigure}[t]{.6\textwidth}
    \centering
    \resizebox{\textwidth}{!}{%
\begin{tikzpicture}[->,>=stealth',auto,node distance=1.25cm,thick,scale=0.9,every node/.style={scale=0.85}]
\tikzstyle{every state}=[text=black, font=\scriptsize]
\tikzstyle{input}=[rectangle,fill=red!30,minimum size=6.5mm]
\tikzstyle{output}=[fill=green!30,minimum size=7.5mm]
\tikzstyle{input char}=[text=Red3]
\tikzstyle{output char}=[text=Green4]

\node[state, initial, initial text={}, input] (i) {1};
\node[state, right=of i, output]              (p) {2};
\node[state, right=of p, input]               (q) {3};
\node[state, right=of q, output]              (s) {5};
\node[state, above=0.75cm of s, output]       (r) {4};
\node[state, right=of s, input, accepting]    (t) {6};
\node[state, above=0.75cm of t, output]       (u) {7};

\path (i) edge            node[input char]                   {$\top, \downarrow{} r$}     (p);
\path (p) edge            node[output char]                  {$r^=$}             (q);
\path (q) edge            node[sloped,pos=0.8,input char]    {$r^{\neq}$} (r);
\path (q) edge            node[input char]                   {$r^=, \downarrow{} r$}    (s);
\path (r) edge            node[sloped,pos=0.2,output char]   {$r^{\neq}$}                 (t);
\path (s) edge            node[output char]                  {$r^=$}                   (t);
\path (t) edge[bend left] node[left,near end,input char]     {$\top$}                     (u);
\path (u) edge[bend left] node[right,near start,output char] {$\top$}                     (t);
\end{tikzpicture}
    }
  \caption{A similar \dra, suprisingly realisable.}%
  \label{fig:cornerCaseDRA2}
  \end{subfigure}%
  \caption{Pathological \dra specifications.}%
  \label{fig:cornerCaseDRA}
\end{figure}

We reduce the realisability of \draa-specifications to solving a
finite parity game. To ease its construction, we first need to confer additional
properties to the specification automaton.

A \ra $A$ is said to be \emph{locally concretisable} if for every finite sequence
of transitions $\rho = t_1\dots t_n$, for every finite data word $w\in
\findatawords$ such that $\rho$ is a partial run of $A$ on $w$, we have that for
all transitions $t \in \delta$ which are compatible with $\rho$ (i.e.\ such that
the source state of $t$ is equal to the end state of $\rho$), there exists $d
\in \data$ such that $\rho t$ is a partial run of $A$ on $w d$. Note in
particular that when $\rho$ is not a partial run, such condition trivially
holds.

We say that a \ra $A$ is in \emph{good form} if
\begin{enumerate}
\item\label{itm:locallyConcret} it is locally concretisable
\item\label{itm:completeInput} it is complete on its input states
\item\label{itm:testsMaxConsist} its tests $\phi$ are maximally consistent conjunctions of atoms
\item\label{itm:NoOutputAssgn} any transition $t$ whose test is different from $\bigwedge_{r\in
    R} r^{\neq}$ does not conduct an assignment ($\asgn = \varnothing$)
\end{enumerate}
\begin{lem}\label{lem:trim}
  For all \ra $A$, there exists an equivalent \ra $A'$ in good form with
  exponentially many more states and transitions, and the same number
  of priorities and registers. Moreover if $A$ is a \draa, so is $A'$.
\end{lem}
\begin{proof}
  Let $A = (\Sigma, \data, Q, q_0, \delta, R, c)$ be a $\ra$. First, we can
  assume that $A$ is complete on its input states: add two sink states $s_\inp$
  and $s_\outp$ with transitions $(s_\inp,\sigma_\inp,\top,\varnothing,s_\outp)$
  and $(s_\outp,\sigma_\outp,r^=,\varnothing,s_\inp)$ for all $\sigma_\inp\in
  \inalph, \sigma_\outp\in\outalph, r \in R$, each with odd priority $c(s_\inp)
  = c(s_\outp) = 1$. Then, for all $q_\inp \in Q_\inp$, and all finite label
  $\sigma_\inp \in \inalph$, add a transition $q_\inp \xrightarrow{\sigma_\inp,
    \psi, \varnothing} s_\outp$ where $\psi = \neg \bigvee_{q_\inp
    \xrightarrow{\sigma_\inp, \phi, \asgn} q_\outp} \phi$ is a test which is
  satisfied by a data if and only if such data satisfies no other possible test.
  This does not affect determinism nor the recognised language (as each added
  state has odd priority), and preserves the fact of being \textsf{ido}.

  Now, we enrich the states with information on the equalities between registers
  in the current register valuation.
  Formally, we define constraints\footnote{The notion of \emph{constraint} is
    pervasive in the study of registers automata, e.g.\ to recognise the
    projection over finite labels.} as equivalence relations on $R$. In the
  following, we denote by $\textsf{ER}(R)$ the set of equivalence relations on
  $R$. Given a valuation $\tau$ of registers in $R$, we can associate to it an
  equivalence relation on $R$ in the natural way (two registers $r,r'\in R$ are
  equivalent iff $\tau(r)=\tau(r')$). We denote it by $[\tau]$. We use the
  letter $C$ to denote an element of $\textsf{ER}(R)$, and we call it a
  constraint.

  We let $A' = (\Sigma, \data, Q', q'_0, \delta', R, c')$ be defined as follows:
  \begin{itemize}
  \item $Q' = Q \times \textsf{ER}(R)$
  \item $q'_0 = \left( q_0,[\tau_0^R] \right)$
  \item $c'(q,C) = c(q)$, for every $(q,C)\in Q \times \textsf{ER}(R)$
  \item $\delta'$ will be defined in the sequel.
  \end{itemize}

\noindent
  Given a constraint $C$, and a set $E\subseteq R$ corresponding to an
  equivalence class of $C$, we define a test corresponding to a maximally
  consistent conjunction of equalities and inequalities: $\alpha_E =
  \bigwedge_{r\in E}r^= \wedge \bigwedge_{r\not\in E} r^{\neq}$. A data value
  satisfies this test iff it is equal to the (common) value stored in registers
  of $R$. We also consider the test $\alpha_\varnothing = \bigwedge_{r\in
    R}r^{\neq}$ which corresponds to the case of a fresh data value, i.e.\ a data
  value distinct from all the values stored in registers.

  Consider a transition $(p,\sigma,\phi,\asgn,q)\in \delta$. Given a formula
  $\alpha_E$ as defined above, one can decide whether the formula $\alpha_E
  \Rightarrow \phi$ is valid or not. If this is the case, then we add the
  following transition to $\delta'$:
  \[
  (p,C) \xrightarrow{\sigma, \alpha_E, \asgn} (q,C')
  \]
  where $C'$ is defined as follows: two registers $r,r'$ are in relation with
  respect to $C'$ if and only if one of the following cases holds:
  \begin{itemize}
  \item they are in relation in $C$, and not in $\asgn$
  \item they are both in $\asgn$
  \item $r$ belongs to $E$ and $r'$ belongs to $\asgn$, or vice versa.
  \end{itemize}
  First, observe that since $A$ is complete on its input states, so is $A'$ and
  property~(\ref{itm:completeInput}) holds. Moreover, by definition, $A'$
  satisfies property~(\ref{itm:testsMaxConsist}).

  Now, one can show by induction on the length $n$ of the partial run that every
  partial run $\rho=t_1 \dots t_n$ of $A'$ over some finite data word $w \in
  \findatawords$ reaching some configuration $((p,C),\tau)$ satisfies
  $C=[\tau]$. Thus, for every run of $A'$, by denoting
  ${\{((q_i,C_i),\tau_i)\}}_{i \in \N}$ its sequence of configurations, we have
  $C_i = [\tau_i]$.

  Additionally, for each run of $A$, we can build a run of $A'$ in a
  deterministic manner: let $\rho = t_1 t_2 \dots$ be a run of $A$ over some
  data word $w = (\sigma_1,d_1) (\sigma_2,d_2) \dots$, where for all $i \in \N$,
  $t_{i+1} = q_i \xrightarrow{\sigma_i, \tst_i, \asgn_i} q_{i+1}$ and let
  ${\{(q_i,\tau_i)\}}_{i \in \N}$ be its sequence of configurations.
  Correspondingly, let $\rho' = t'_1 t'_2 \dots$, where for each $i \in \N$
  $t'_{i+1} = (q_i,C_i) \xrightarrow[A']{\sigma_i,\alpha_{E_i},\asgn_i}
  (q_{i+1},C_{i+1})$, with $C_i = [\tau_i]$ and $E_i = \{r \in R \mid \tau_i(r)
  = d_i\}$. Then, again by induction, we can show that $\rho'$ is a run of $A'$
  over $w$, whose sequence of configurations is ${\{((q_i,C_i),\tau_i)\}}_{i \in
    \N}$. Moreover, $\rho'$ is accepting if and only if $\rho$ is accepting,
  since $c'(q_i,C_i) = c(q_i)$. Reciprocally, every run $\rho'$ of $A'$ can be
  projected to a run of $A$ by removing the $C_i$, and this preserves
  acceptance. Overall, $L(A) = L(A')$.

  Now, let $\rho = t_1 \dots t_n$ be a partial run of $A'$ over some finite data
  word $w \in \findatawords$ ending in some configuration $((q,C),\tau)$; recall
  that $C = [\tau]$. Let $t = q \xrightarrow{\sigma,\alpha_E,\asgn} q'$ be a
  transition compatible with $\rho$, i.e.\ such that $q$ is the end state of
  $\rho$. If $E = \varnothing$, then $\alpha_E = \bigwedge_{r \in R} r^{\neq}$,
  so any $d \in \data \backslash \tau(R)$ (where $\tau(R)$ denotes the image of
  $R$ by $\tau$) is such that $\tau, d \models \alpha_E$. If $E \neq
  \varnothing$, then by construction $E$ corresponds to an equivalence class of
  $C$, so $\forall r,r' \in E, \tau(r)=\tau(r')$ and $\forall r \in E, \forall
  r' \notin E, \tau(r) \neq \tau(r')$. Thus, by letting $d = \tau(r)$ for some
  $r \in E$ (its choice does not matter), we have that $\rho t$ is a partial run
  of $A'$ over $w d$. Overall, $A'$ is locally concretisable, i.e.
  property~(\ref{itm:locallyConcret}) holds.

The last step concerns property~(\ref{itm:NoOutputAssgn}). Intuitively, the idea is that if the data read corresponds to a data stored in some register, then the assignment can be replaced by keeping in memory
a relation between registers. This idea is merely an adaptation of the conversion from register
  automata (``$M$-automata'', in their terminology) to finite-memory
  automata~\cite{Kaminski:1994:FA:194527.194534}. The states can be enriched with
  the right information to deal with these additional relations.
\end{proof}

In order to solve the unbounded register synthesis problem, we
resort to a synthesis problem for data-free specifications.
In that framework, when specifications are described by means of parity
automata, synthesis problems can be solved using reductions
to parity games. We thus quickly recall the notion of parity game.
For a complete presentation, we refer the reader to~\cite{Apt:2011:LGT:1972520}.

A two-player parity game is given as a finite graph, in which
vertices are partitioned among the two players, together with an initial vertex.
A colouring function associates
with each vertex an integer. It is used to define the winning plays as follows: a play is winning iff the maximum colour appearing infinitely often is even.

In the sequel, we will use the parity game associated with a \dra $A$,
which is denoted as $G_A$. It is
is defined as follows: its set of vertices is exactly that of $A$. Player Adam
owns input vertices, and the associated input transitions, while player Eve
owns output vertices/transitions. The colouring function is that of $A$,
and the initial vertex is the initial state of $A$.
\begin{prop}%
  \label{prop:implDRAnew}
  Let $A$ be a \draa in good form. Then, the following are equivalent:
  \begin{enumerate}
  \item\label{itm:trdImpl} $L(A)$ is realisable by a register transducer with
    as many registers as $A$
  \item\label{itm:impl} $L(A)$ is realisable by an
    implementation\footnote{Recall that implementations are defined in
      \autoref{sec:specImplReal}.} $I : {(\inalph \times \data)}^* \rightarrow
    \outalph \times \data$
  \item\label{itm:winStrat} Eve wins the parity game $G_A$ associated with $A$
  \end{enumerate}
\end{prop}

\begin{proof}
We start with a preliminary remark on \draa.
As $A$ is a \draa, every output
transition has a test with at least one equality constraint ($r^=$ for some $r$),
and thus, as $A$ is in good form (property $(4)$), the assignment of output
transitions are all empty.
Note that~\ref{itm:trdImpl} $\Rightarrow$~\ref{itm:impl} is immediate.

\subsection*{\texorpdfstring{From the parity game $G_A$ to the realisability of $L(A)$:~\ref{itm:winStrat} $\Rightarrow$~\ref{itm:trdImpl}}{From the parity game GA to the realisability of L(A): 3 => 1}} 
Assume Eve wins the game $G_A$. Parity games admit memoryless strategies, i.e.
strategies whose actions only depend on the current state of the game. We can
thus consider a memoryless winning strategy for Eve, which we denote by a
mapping $\chi$ from output vertices to output edges of the game, i.e.\ from
output states to output transitions of $A$.

We now detail how we define from $\chi$ a register transducer $T_\chi$ with
$R^A$ as set of registers:
\begin{itemize}
\item States are those of $A$
\item The initial state is that of $A$
\item Transitions are defined as follows. Consider some input state $p$ and some
  transition $t_\inp$ from $p$ to $q$. By definition of $A$, $q$ is an output
  state, and we let $t_\outp = \chi(q)$ be the transition given by Eve's
  strategy.

  We write $t_\inp=(p,\sigma,\phi,\asgn,q)$ and
  $t_\outp=(q,\sigma',\phi',\asgn',q')$. Thanks to our initial comment on the
  form of output transitions of $\draa$ in good form, there exists a register
  $r$ appearing with an equality constraint in the test $\phi'$ of the
  transition $t_\outp$, and we have $\asgn'=\varnothing$. Then, we add to
  $T_\chi$ the transition $p\xrightarrow{\sigma,\phi \mid \asgn,\sigma',r}q'$.
\end{itemize}

\noindent
Observe that $T$ is indeed a register transducer as for each state $p$, it only
uses transitions outgoing from $p$ in $A$, hence it is deterministic as $A$ was.

We claim that $T_\chi$ realises $L(A)$. Consider some input data word, and the
behaviour of $T_\chi$ on this data word. As $A$ is in good form, it is complete
on its input states. This entails that this run is infinite. It corresponds to a
play in $G_A$ compatible with Eve's strategy $\chi$. As $\chi$ is a winning
strategy, this implies that the run is accepting, hence corresponds to some
accepting run of $A$, yielding the result.

\subsection*{\texorpdfstring{From the realisability of $L(A)$ to the parity game $G_A$:~\ref{itm:impl} $\Rightarrow$~\ref{itm:winStrat}}{From the realisability of L(A) to the parity game GA: 2 => 1}} 
Assume that $L(A)$ is realisable by an implementation $I : {(\inalph \times
\data)}^* \rightarrow \outalph \times \data$. We let $f_I :
\datawords(\inalph,\data) \rightarrow \datawords(\outalph,\data)$ be the
function it implements, and naturally extend it to finite words: for $w_\inp \in
\findatawords(\inalph,\data), f_I(w_\inp) = I(w_\inp[1]) I(w_\inp[1:2]) \dots
I(w_\inp[1:\size{w_\inp}])$. Let us build from $I$ a winning strategy $\chi_I$
in $G_A$, with memory ${(\inalph \times \data)}^* \times (Q_S \times
\data^{R_A})$.

We define $\chi_I$ by induction, and show that when $\chi_I$ is
in memory state $(w_\inp,(q,\tau))$, the finite sequence of transitions
constructed so far is a partial run of $A$ over $\interl{w_\inp}{f_I(w_\inp)}$
ending in configuration $(q,\tau)$. Initially, $\chi_I$ has memory
$(\varepsilon, (q_0,\tau_0))$.

Now, assume $\chi_I$ is in state $(w_\inp,
(q,\tau))$, and Adam just played $(\sigma_\inp, \tst, \asgn)$. Then, Eve picks
some data $d_\inp \in \data$ such that $\tau, d_\inp \models \tst$. Such data
exists since $A$ is locally concretisable and the finite sequence of transitions
constructed so far is the partial run over some data word. Let $(q'', \tau'')$
be the successor configuration of $(q,\tau)$ in $A$ on reading $d_\inp$, i.e.
$(q,\tau) \xrightarrow[A]{\sigma_\inp,d_\inp} (q'',\tau'')$, and let $w'_\inp = w_\inp (\sigma_\inp, d_\inp)$. Now, let
$(\sigma_\outp, d_\outp) = I(w'_\inp)$.
Correspondingly, let $t_\outp$ be the transition taken from $(q'',\tau'')$ on
reading $(\sigma_\outp, d_\outp)$, i.e.\ such that $(q'', \tau'')
\xrightarrow[t_\outp]{\sigma_\outp, d_\outp} (q',\tau')$. Such transition
exists: let $w \in \datawords(\inalph,\data)$ be some infinite suffix
that we append to $w'_\inp$. Since $I$ is an implementation, $f_I$ is total and
we know that $\interl{w'_\inp w}{f_I(w'_\inp w)} \in L(A)$, which means that
$\interl{w'_\inp w}{f_I(w'_\inp w)}$ admits an accepting run in $A$.
In particular, its prefix $\interl{w'_\inp}{f_I(w'_\inp)}$ admits a partial
run in $A$, and its last transition is $t_\outp$ (such partial run is unique
since $A$ is deterministic).

Then, Eve plays $t_\outp$ in $G_A$ and updates her memory to $\left(w'_\inp,
  (q',\tau')\right)$. The invariant indeed holds, as the play constructed so far
is a partial run of $A$ over $\interl{w'_\inp}{f_I(w'_\inp)}$ ending in configuration
$(q',\tau')$.

$\chi_I$ is indeed a strategy, as it is defined for any possible sequence of
actions of Adam. It remains to show that it is winning. Let $\rho$ be a play
consistent with $\chi_I$, which is also a run of $A$ by definition of $G_A$. We
need to show that $\rho$ is accepting. We define $w \in {(\inalph \times
\data)}^\omega$ as $w[i] = w_\inp^i[i]$, where $w_\inp^i$ is the input word
stored in memory at step $i$ of the play (i.e.\ such that $\chi_I$ is in state
$(w_\inp^i, (q_i,\tau_i))$ for some $(q_i,\tau_i)$ after receiving $i$ actions
of Adam). We then know that for all $i \in \N$, $\rho[:i]$ is a partial run of
$A$ over $\interl{w[:i]}{f_I(w[:i])}$, so $\rho$ is a run of $A$ over
$\interl{w}{f_I(w)}$. Since $I$ is an implementation, such run is accepting,
i.e.\ satisfies the parity condition, which means that $\rho$ also satisfies the
parity condition; it is thus winning. As a consequence, $\chi_I$ is a winning
strategy in $G_A$.
\end{proof}

\begin{thm}\label{thm:unboundedDRA} {\normalfont$\real(\draa, \rt)$} is {\normalfont\textsc{ExpTime-c}}.
\end{thm}

\begin{proof}
  First, we put $A$ in good form thanks to Lemma~\ref{lem:trim}, resulting in
  some \draa $B$ exponentially bigger. Then, by
  Proposition~\ref{prop:implDRAnew}, it suffices to solve the parity game $G_B$.
  It is well-known to be possible in time $O(n^d)$ where $n$ is the number of
  states and $d$ the number of priorities. If $n_A$ denotes the number of states
  of $A$ and $d$ its number of priorities, then $B$ has $n_A \cdot 2^{|R|^2}$
  states and the same number of priorities $d$, hence checking the realisability
  of $A$ can be done in time $O(n_A^d \cdot 2^{d \cdot |R|^2})$, which is
  exponential with respect to the size of the input.

  \subsubsection*{Hardness} The following proof is an adaptation of the one
  establishing \textsc{PSpace}-hardness of the nonemptiness problem for \dra
  presented in~\cite[Theorem 5.1]{DBLP:journals/tocl/DemriL09}. Here, we use the
  input part to simulate universal transitions, and the output part to simulate
  nondeterministic ones, hence simulating alternation, which yields an
  \textsc{ExpTime} lower bound.

  Thus, we reduce from the halting problem of alternating Turing machines over
  a binary alphabet with a linearly bounded tape. An alternating Turing machine
  is a tuple $\tmm = \langle Q, q_i, \delta \rangle$, where:
  \begin{itemize}
  \item $Q$ is a finite set of \emph{states}, partitioned into
    \emph{existential} ($Q_\exists$) and \emph{universal} ($Q_\forall$) states:
    $Q = Q_\exists \uplus Q_\forall$, where $q_i \in Q_\forall$ is the
    \emph{initial} state
  \item $\delta: Q \times \{0,1\} \rightarrow 2^{Q \times \{0,1\} \times
      \{-1,1\}}$ is the \emph{transition function}.
  \end{itemize}
  A configuration of $\tmm$ is then a triple $ c=(q,i,w)$, where $q \in Q$
  is the machine state, $i \in \{0, \dots, \size{\tmm}-1\}$ is the head
  position, and $w \in {\{0,1\}}^{\size{\tmm}}$ is the tape content. It is
  \emph{existential} if $q \in Q_\exists$ and \emph{universal} if $q \in Q_\forall$. A configuration
  $(q',i',w')$ is a successor of $(q,i,w)$ if there exists $(p,
  a, m) \in \delta(q,w[i])$, $p=q'$, $i'= i+m \in \{0, \dots, \size{\tmm} - 1\}$
  and $w'$ is such that $\forall j \neq i$, $w'[j] = w[j]$ and $w[i] = a$.
  $t = q \xrightarrow{w[i],a,m} p$ is called the \emph{associated
    transition}. A \emph{run} of $\tmm$ is then a tree whose nodes are
  configurations and whose branches can be finite or infinite, rooted in the initial
  configuration $(q_i, 0, 0^{\size{\tmm}})$, and whose nodes satisfy the following properties:
  \begin{enumerate}
  \item If the node is an existential configuration $c_\exists$, then it has exactly one child, which is a successor configuration of $c_\exists$.
  \item If the node is a universal configuration $c_\forall$, then its children
    are all its successor configurations.
  \end{enumerate}
  Note that a branch is finite if and only if it ends in a universal
  configuration with no successor. The machine $\tmm$ \emph{halts} if it admits
  a run which is a finite tree (i.e.\ whose branches all end in a universal
  configuration with no successors). The following problem is
  \textsc{ExpTime}-hard~\cite{Chandra:1981:ALT:322234.322243}: given an
  alternating Turing machine $\tmm$, decide whether $\tmm$ halts.

  Finally, a computation is a finite sequence of successive configurations (i.e.
  a finite path in a run). Let
  $(q_0, i_0, w_0) \dots$  $(q_n, i_n, w_n)$ be a computation of $\tmm$, and
  $t_0 \dots t_{n-1}$ the sequence of associated transitions. We encode such
  computation by the following data word over the alphabet $Q \uplus \delta
  \uplus \{-\}$:
  \begin{equation*}
    (-,d_0) (-,d_1) a_0^0 a_1^0 \dots a_{\size{\tmm}-1}^0 t_0 a_0^1 a_1^1 \dots a_{\size{\tmm}-1}^1 t_1 \dots t_{n-1} a_0^n a_1^n \dots a_{\size{\tmm}-1}^n
  \end{equation*}
  where $d_0 \neq d_1 \in \data$ are two distinct data respectively encoding
  letters $0$ and $1$, and we have $\lab(a_l^k) = q_k$ if $l = i_k$ and
  $\lab(a_l^k) = -$ otherwise. Then, $\dt(a_l^k) = d_0$ if $w_k[l] = 0$ and
  $\dt(a_l^k) = d_1$ if $w_k[l] = 1$. $\dt(t_k)$ does not matter.

  Now, as in~\cite{DBLP:journals/tocl/DemriL09}, we can construct a \dra
  $A_{\tmm}$ which accepts a data word iff it has a prefix that encodes a
  computation of $\tmm$ from the initial state to a state with no successor.
  Indeed, the transitions are part of the input, so they do not have to be
  guessed: neither nondeterministic nor universal branching is needed here (they
  will respectively be simulated by the output and input player). For
  completeness, we describe the construction: $A_{\tmm}$ has memory $Q$,
  along with an $\size{\tmm}$-bounded counter $l$ to keep track of the position
  of the reading head in $w_k$, a variable $i$ taking its values in $\{0, \dots,
  \size{\tmm}-1\}$ used to store the value of $i_k$ and a variable $t$ taking
  its values in $\delta$ to memorise $t_k$; which overall yields a
  $O(\size{\tmm}^4)$ memory. Its finite alphabet is $Q \uplus \delta \uplus
  \{-\}$, and it has $\size{\tmm} + 2$ registers: $r_0$ and $r_1$ respectively
  store $d_0$ and $d_1$, and, for all $0 \leq l < \size{\tmm}$, $r'_l$
  successively stores the different values of $w_k[l]$ for $0 \leq k \leq n$.
  Then, a run of $A_{\tmm}$ is as follows: initially, $A_{\tmm}$ stores $d_0$
  and $d_1$, while checking that they are distinct. Then, it checks that $w_0 =
  0^{\size{\tmm}}$. To check successorship, while maintaining the invariant that at
  any step $k$, $r'_l$ contains $w_k[l]$, the automaton, when reading $t_k = q
  \xrightarrow{c,a,m} p$, checks that $q = q_k$ (it was stored as the target of
  $t_{k-1}$), $c = w_k[i_k]$ (i.e.\ that $r'_{i_k}$ contains $d_c$), and
  updates the value of $i_k$ to $i_{k+1} = i_k + m_k$, while checking that $i_k \in \{0,
  \dots, \size{\tmm}-1\}$. Then, with the help of its registers and its counter
  $l$, it checks that $w_{k+1}[l] = w_k[l]$ for all $l \neq i_{k+1}$, and that
  $w_{k+1}[i_{k+1}] = d_a$.

  From such automaton, by adding $\#$s to enforce the alternation between input
  and output, we can build a specification automaton such that the input player
  provides the encoding of the successive configurations, and resolves the
  universal branching, and the output player has to resolve nondeterminism (i.e.
  chooses which nondeterministic transition to take).
  Then, if the input player can force the computation to go on ad infinitum, he
  wins, otherwise (if either the provided encoding is not correct, or if the
  computation is finite), the output player wins. Formally:
  \begin{align*}
    S = &\{(-,d_0) \# (-,d_1) \# \interl{c_0}{\#^{\size{\tmm}}} t_0 \# \interl{c_1}{\#^{\size{\tmm}}} \# t_1 \interl{c_2}{\#^{\size{\tmm}}} t_2 \# \dots \interl{c_n}{\#^{\size{\tmm}}} \#^\omega \mid \\
        &\phantom{\{}d_0 \neq d_1 \text{ and }c_0t _0 c_1 t_1 c_2 t_2 \dots t_{n-1}c_n\text{ is the encoding of a computation of }\tmm \} \\
    {}\cup{} & \left\{\interl{w}{w'} \;\middle|\; \begin{array}{c} \text{there exists a prefix of $w$ which is not} \\ \text{the encoding of a computation of }\tmm \end{array} \right\} \\
    {}\cup{} & \{\interl{(-,d_0) \# (-,d_1) w}{w'} \mid d_0 = d_1\} 
  \end{align*}
  The data corresponding to the $\#$ and $t_i$ do not matter, and are not
  depicted. Note that the even (i.e.\ universal) transitions are picked by the
  input player, while the odd (i.e.\ nondeterministic) transitions are picked by
  the output player.

  Now, if $\tmm$ halts, $A$ admits an implementation, which behaves as follows:
  it first checks that the $d_0$ and $d_1$ given as input are indeed distinct.
  Then, it checks on-the-fly that the given input is indeed an encoding of the
  initial configuration, while outputting $\#$s. It then checks that $c_1$ is
  indeed a successor of $c_0$ following $t_0$, again while outputting $\#$s.
  Then, if it receives a $\#$  as input, it picks some $t_1$ which is a witness
  that $c_0$ is indeed accepting, and so on. If, at some point, the given input
  is not a valid encoding, then it behaves arbitrarily (e.g.\ by outputting only
  $\#$s).

  Conversely, if $\tmm$ does not halt, then, by choosing an input whose
  universal transitions are witnesses that $c_0$ is not accepting, then either
  the implementation provides some non-admissible output at some point, or the
  computation goes ad infinitum, which breaks the specification.

  For readers familiar with game-theoretic formulations, winning strategies in
  the synthesis game of $A_{\tmm}$ are in one-to-one correspondence with halting
  runs of $\tmm$.
\end{proof}

As a consequence of the fact that if a \draa is realisable, then it is so by
a register transducer with the same number of registers, we obtain
the following corollary:

\begin{cor}%
\label{cor:boundedDRA}
 Let $k\ge r$ be two integers. We denote by $\draa[r]$ the class of \draa with $r$ registers.
 {\normalfont$\real(\draa[r], \rt[k])$} is in {\normalfont\textsc{ExpTime}}.
\end{cor}

\section{Bounded Synthesis: A Generic Approach}

In this section, we study the setting where target implementations are register
transducers in the class $\rt[k]$, for some $k \geq 0$ that we now fix for the
whole section. For the complexity analysis, we assume $k$ is given as input, in
unary. Indeed, describing a $k$-register automaton in general requires $O(k)$
bits, and not $O(\log k)$ bits. We prove the decidable cases of the first line
of Table~\ref{tbl:results} (page~\pageref{tbl:results}), by reducing the
problems to realisability problems for data-free specifications.

\subsection{Abstract Actions}\label{subpar:abstractActions}
We let $R_k = \{1,\dots,k\}$ be a set of $k$ registers.
Our aim is to reduce the problem to a finite alphabet problem. First,
since the set of test formulas over $R_k$ is infinite and there are
doubly exponentially many non-equivalent formulas over $R_k$, we
rather synthesise transducers whose tests are maximally consistent
conjunctions of atoms of the form $r^=$ or $r^{\neq}$. Such
conjunctions can be identified as subsets of $R_k$ in a natural way,
e.g.\ for $k=3$,  the test $r_1^=\wedge r_2^{\neq} \wedge r_3^=$ is
identified with the set $\{1,3\}$. We call them explicit tests and
denote them by the capital letter $E$. An
explicit test $E\subseteq R_k$ is converted into the (implicit) test
$\phi_E = \bigwedge_{r\in E} r^=\wedge \bigwedge_{r\not\in E} r^{\neq}$.
Explicit tests  are for instance used in~\cite{DBLP:conf/csl/Segoufin06}.

We let $\Tst_k = \Asgn_k = 2^{R_k}$. The finite input actions are
$A_\inp^k = \inalph\times \Tst_k$ which corresponds to picking a
label and a test over the $k$ registers, and the output actions are
$A_\outp^k = \outalph\times \Asgn_k \times R_k$, corresponding to
picking some output symbol, some assignment and some register whose
content is to be output.

An alternating sequence of actions
$\overline{a} =
(\sigma_\inp^1,\tste_1)(\sigma_\outp^1,\asgn_1,r_1)\dots
\in {(A_\inp^{k}A_\outp^{k})}^\omega$ abstracts a set of relational data
words of the form
$w = (\sigma_\inp^1,d_\inp^1)(\sigma_\outp^1,d_\outp^1)\dots \in\relwords(\inalph,\outalph, \data)$ via a
compatibility relation that we now define. We say that $w$ is
\emph{compatible} with $\overline{a}$ if there exists a sequence of
register configurations $\tau_0\tau_1\dots \in {(R_k\rightarrow
\mathcal{D})}^\omega$ such that $\tau_0 = \tau_0^{R_k}$ and
for all $i\geq 1$, $\tau_i,d_{\inp}^i\models
\tste_i$, $d_\outp^i = \tau_i(r_i)$ and $\tau_{i+1} =
\text{next}(\tau_i,d_\inp^i,\asgn_i)$. In other words, $w$ is \emph{compatible}
with $\overline{a}$ if there exists some $k$-register transducer and a run
$\rho = t_0 t_1 \dots$ such that for all $i$, $t_i$ is of the form $t_i = q_i
\xrightarrow{\sigma_\inp^i, \tste_i \mid \sigma_\outp^i, \asgn_i, r_i} q_{i+1}$
for some $q_i, q_{i+1} \in Q_T$.
Note that this sequence
is unique if it exists. We denote by
$\textsf{Comp}(\overline{a})$
the set of relational data words
compatible with $\overline{a}$.
Given a specification $S$, we let $W_{S,k} = \{
\overline{a}\mid \textsf{Comp}(\overline{a})\subseteq S\}$. The set
$W_{S,k}$ is then a specification over the finite input and output alphabets
$A_\inp^k$ and $A_\outp^k$.

\begin{thm}[Transfer]\label{thm:infinite2finite}
    Let $S$ be a data word specification. The following are
    equivalent:
\begin{enumerate}
    \item $S$ is realisable by a transducer with $k$
    registers.
  \item The (data-free) word specification
    $W_{S,k}$ is realisable by a (register-free) finite transducer.
  \end{enumerate}
\end{thm}

\begin{proof}
  Let $T$ be a transducer with $k$ registers realising $S$. The tests
  of $T$ are implicit tests, so in a first step we explicit them,
  possibly by adding new transitions to $T$. Formally, a transition
  $q\myxrightarrow[T]{\sigma_\inp,\phi\mid \sigma_\outp,\asgn,r} q'$ is
  replaced by all the transitions $q\myxrightarrow[T]{\sigma_\inp,E\mid
    \sigma_\outp,\asgn,r} q'$ for all $E\subseteq R_k$ such that
  $\phi_E\Rightarrow \phi$ is true. The resulting
  transducer can be seen as a finite transducer $T'$ over input alphabet $A_\inp^k$ and
  output alphabet $A_\outp^k$. Moreover, since the transition function of $T$ is
  complete, it is also the case of $T'$ (this is required by the definition of
  transducer defining implementations).


  Let us show that
  $W_{S,k}$ is realisable by $T'$, i.e.\ $L(T')\subseteq W_{S,k}$. Take a
  sequence $\overline{a} = a_1e_1a_2e_2\dots \in L(T')$. We show that
  $\textsf{Comp}(\overline{a})\subseteq S$. Let $w\in
  \textsf{Comp}(\overline{a})$. Then, there exists a run $q_0q_1q_2\dots$ of
  $T'$ on $\overline{a}$ since $\overline{a}\in L(T')$. By definition of
  compatibility for $w$, there exists a sequence of register configurations
  $\tau_0\tau_1\dots \in {(R_k\rightarrow \data)}^\omega$ satisfying the
  conditions in the definition of compatibility. From this we can deduce that
  $(q_0,\tau_0)(q_1,\tau_1)\dots$ is an initial sequence of configurations of
  $T$ over $w$, so $w\in L(T)$. Finally, $T$ realises $S$, and therefore
  $L(T)\subseteq S$.

  Conversely, suppose that $W_{S,k}$ is realisable by some finite transducer
  $T'$ over the input (output) alphabets $A_\inp^k$ ($A_\outp^k$). Again, the
  transducer $T'$ can be seen as a transducer $T$ with $k$ registers over
  data words with explicit tests.
  We show that $T$ realises $S$, i.e., $L(T)\subseteq S$. Let $w\in L(T)$. The
  run of $T$ over $w$ induces a sequence of actions $\overline{a}$ in ${(A_\inp^k
  A_\outp^k)}^\omega$ which, by definition of compatibility, satisfies $w\in
  \textsf{Comp}(\overline{a})$. Moreover, $\overline{a}\in L(T')$. Hence, since
  $T'$ realises $W_{S,k}$, we get $\textsf{Comp}(\overline{a})\subseteq S$, so
  $w\in S$, concluding the proof.
\end{proof}

\subsection{\texorpdfstring{The case of $\ura$ specifications}{The case of URA specifications}}%
\label{sec:boundedURA}

In this section,
we show that for any $S$ a data word specification given as some $\ura$, the
language $W_{S,k}$ is effectively $\omega$-regular, entailing the decidability
of $\real(\ura, \rt[k])$, by Theorem~\ref{thm:infinite2finite} and the
decidability of (data-free) synthesis. Let us first prove a series of
intermediate lemmas.

We define an operation $\otimes$ between relational data words $w\in
\relwords(\inalph,\outalph,\data)$ and sequences of actions $\overline{a}\in
{(A_\inp^{k}A_\outp^{k})}^\omega$ as follows: $w\otimes \overline{a}\in
\relwords(A_\inp^k,A_\outp^k,\data)$ is defined only if for all $i\geq 1$,
$\lab(w[i]) = \lab(\overline{a}[i])$ where $\lab(\overline{a}[i])$ is the first
component of $\overline{a}[i]$ (a label in $\inalph\cup \outalph$), by
$(w\otimes \overline{a})[i] = (\overline{a}[i],\dt(w[i]))$. Note that such
operation is always defined when $w \in \textsf{Comp}(\overline{a})$.

\begin{lem}\label{lem:compreg}
  The language $L_{k} = \{ w\otimes \overline{a}\mid
  w\in\textsf{Comp}(\overline{a})\}$ is definable by some \nra.
\end{lem}
\begin{proof}
  We define an \nra with $k$ registers which roughly follows the actions it
  reads on its input. Its set of states is $\{q\}\cup \Asgn_{R}$, with initial
  state $q$. In state $q$, it is only allowed to read labelled data in
  $A_\inp^k\times \data$. On reading $(\sigma_\inp, \tst, d)$, it guesses some
  assignment $\asgn$, performs the test $\tst$ and the assignment $\asgn$ and
  goes to state $\asgn$. In any state $\asgn \in \Asgn_{R}$, it is only allowed
  to read labelled data of the form $(\sigma_\outp,\asgn,r,d)$, for which it
  tests whether $d$ is equal to the content of $r$. It does no assignment and
  moves back to state $q$. All states are accepting (i.e.\ have parity $0$).
  Such \nra has size $O(2^{k^2})$.
\end{proof}

Let $S$ be a specification defined by some \ura $A_S$ with set of states $Q$.
The following subset of $L_k$ is definable by some \nra, where $\overline{S}$
denotes the complement of $S$:

\begin{lem}\label{lem:compregS}
  The language $L_{\overline{S},k} = \{ w\otimes \overline{a}\mid
  w\in\textsf{Comp}(\overline{a})\cap \overline{S}\}$ is definable by some \nra.
\end{lem}
\begin{proof}
  Since $S$ is definable by the \ura $A_S$, $\overline{S}$ is \nra-definable
  with $\overline{A_S}$, a copy of $A_S$ with colouring function $\overline{c} :
  q \mapsto c(q) + 1$, interpreted as an \nra. Let $B$ be some \nra defining
  $L_k$ (it exists by Lemma~\ref{lem:compreg}). It now suffices to take a
  product of $A_{\overline{S}}$ and $B$ to get an \nra defining
  $L_{\overline{S},k}$.
\end{proof}

Given a data word language $L$, we denote by $\lab(L) = \{ \lab(w)\mid w\in L\}$
its projection on labels. The language $W_{S,k}$ is obtained as the complement
of the label projection of $L_{\overline{S},k}$:

\begin{lem}\label{lem:abstractspec}
  $W_{S,k} = \overline{\lab(L_{\overline{S},k})}$.
\end{lem}
\begin{proof}
  Let $\overline{a}\in {(A_\inp^{k}A_\outp^{k})}^\omega$. Then, $\overline{a} \notin
  W_{S,k} \Leftrightarrow \Comp(\overline{a}) \not \subseteq S \Leftrightarrow
  \exists w \in \relwords, w \in \Comp(\overline{a}) \cap \overline{S}
  \Leftrightarrow \exists w \in \relwords, w \otimes \overline{a} \in
  L_{\overline{S}, k} \Leftrightarrow \overline{a} \in
  \lab(L_{\overline{S},k})$.
\end{proof}

We are now able to show the regularity of $W_{S,k}$.

\begin{lem}\label{lem:WRegularURA}
  Let $S$ be a data word specification, $k\geq 0$. If $S$ is definable by some
  $\ura$ with $n$ states and $r$ registers, then $W_{S,k}$ is effectively
  $\omega$-regular, definable by some deterministic parity automaton with
  $O(2^{n^2 \cdot 16^{{(r+k)}^2}})$ states and $O(n \cdot 4^{{(r+k)}^2})$
  priorities.
\end{lem}

\begin{proof}
  First, $L_{\overline{S},k}$ is definable by some \nra with $O(2^{k^2}n)$
  states and $O(r+k)$ registers by Lemma~\ref{lem:abstractspec}, obtained as
  product between the \nra $\overline{A_S}$ and the automaton obtained in
  Lemma~\ref{lem:compreg}, of size $O(2^{k^2})$. It is known that the projection
  on the alphabet of labels of a language of data words recognised by some \nra
  is effectively regular~\cite{Kaminski:1994:FA:194527.194534}. The same construction, which is based on
  extending the state space with register equality types, carries over to
  $\omega$-words, and one obtains a nondeterministic parity automaton with
  $O(n \cdot 4^{{(r+k)}^2})$ states and $d$ priorities recognising
  $\lab(L_{\overline{S},k})$. It can be complemented into a deterministic parity
  automaton with $O(2^{n^2 \cdot 16^{{(r+k)}^2}})$ states and $O(n \cdot 4^{{(r+k)}^2})$
  priorities using standard constructions~\cite{journals/lmcs/Piterman07}.
\end{proof}

We are now able to reprove the following result, known
from~\cite{DBLP:conf/atva/KhalimovMB18}:
\begin{thm}\label{thm:boundedURA}
  For all $k\geq 0$, {\normalfont$\real(\ura, \rt[k])$} is in {\normalfont\textsc{2ExpTime}}.
\end{thm}
\begin{proof}
  By Lemma~\ref{lem:WRegularURA}, we construct a deterministic parity automaton
  $P_{S,k}$ for $W_{S,k}$. Then, according to Theorem~\ref{thm:infinite2finite},
  it suffices to
  check whether it is realisable by a (register-free) transducer.
  %
  The way to decide it
  is to see $P_{S,k}$ as a two-player parity game and check whether the
  protagonist has a winning strategy. Parity games can be solved in time
  $O(m^{\log d})$~\cite{Calude:2017:DPG:3055399.3055409} where $m$ is the number
  of states of the game and $d$ the number of priorities. Overall, solving it
  requires doubly exponential time, more precisely in $O(2^{n^3 \cdot 16^{{(r+k)}^2}})$.
\end{proof}

\subsection{The case of test-free \nra specifications}%
\label{sec:boundedNRAtf}

Unfortunately, by Theorem~\ref{thm:boundedNRA}, the synthesis problem for
specifications expressed as $\nra$ is undecidable, even when the number of
registers of the implementation is bounded. And indeed, if we mimic the
reasoning of the previous section, we get that  $L_{\overline{S},k}$ is
definable by a $\ura$, but Lemma~\ref{lem:abstractspec} does not allow
to conclude because:
\begin{prop}
  There exists a data word language $L$ which is \ura-definable and whose string
  projection is not $\omega$-regular.
\end{prop}
\begin{proof}
Consider \[L =
\{(r,d_1) \dots (r,d_n) (g, d'_1) \dots (g,d'_m) {(\#,d)}^\omega \mid \forall i \neq j, d_i \neq
d_j \wedge \forall 1 \leq i \leq n, \exists j, d'_j = d_i\},\] which consists in
a word $w \in r^n$ with pairwise distinct data followed by a word $w' \in g^m$
which contains at least all the data of $w$, and extended with ${(\#,d)}^\omega$
to make it infinite (here, the choice of $d$ does not matter). Such language can
be interpreted as the request-grant specification, restricted to the case where
all requests are made first, and are all made by pairwise distinct clients (plus
a $\#$ infinite padding).
$L$ is recognised by an \ura which, on reading $(r,d_i)$, universally triggers
a run checking that
\begin{enumerate}
\item Once a label $g$ is read, only $g$s are read; and after the last $g$, only
  $\#$ are read (this is an $\omega$-regular property)
\item $(r,d_i)$ does not appear again
\item $(g, d_i)$ appears at least once.
\end{enumerate}
Now, we have $\lab(L) = \{r^n g^m \#^\omega \mid m \geq n\}$, which is not
$\omega$-regular.
\end{proof}
In this section, we consider the class of \nra which do not perform tests on
input data, which we call test-free nondeterministic register automata (\nratf
for short). Such restriction is inspired
from~\cite{DBLP:conf/fossacs/Durand-Gasselin16}, which defines transformations
of data words using MSO interpretations with an MSO origin relation. The MSO
interpretation describes the transformation over the finite alphabet (called the
\emph{string transduction}), as in~\cite{Courcelle:1994:MSD:179203.179210}, while the
MSO origin relation describes the relation between input and output data. Such
relation does not depend on (un)equalities between different input data: it 
uniquely maps each output position to an input position, expressing that the
output data at this position is equal to the corresponding input data. They then
show that such model is equivalent to two-way deterministic transducers with
\emph{data variables}\footnote{Themselves equivalent to one-way streaming string
  transducers with data variables and parameters; such parameters are
  reminiscent of the guessing mechanism described
  in~\cite{DBLP:journals/ijfcs/KaminskiZ10}.}. Such data variables are used to
implement the MSO origin relation: they are registers in which the transducer
can store the input data values and output them, but it is not allowed to
perform any test on the stored data, contrary to our model of register automata.
To define \nratf, we apply the same restriction to \nra: they correspond to
\emph{nondeterministic} \emph{one-way} transducers with data variables. Such
machines can only rearrange input data (duplicate, erase, copy) regardless of
the actual data values (as there are no tests). This way, as stated in
Proposition~\ref{prop:originData}, registers induce an origin relation between
input and output data.

To avoid confusion between the nature of specifications and implementations, we
prefer to define them as register automata, instead of transducers.

\begin{defi}[Test-free register automaton]\label{def:nratf}
  A \nra is \emph{test-free} if:
 \begin{enumerate}

  \item Its input transitions do not depend on equality relations between input
    data: for all $t \in \delta$, if $t = q \xrightarrow{\sigma, \tst, \asgn}
   q'$ is an input transition, then $\tst = \top$.

  \item Its output transitions consist in outputting the content of some
    register: for all $t \in \delta$, if $t = q \xrightarrow{\sigma, \tst,\asgn} q'$ is an output transition, then $\tst = r^=$ for some
    $r\in R$ and $\asgn = \varnothing$.
  \end{enumerate}

\end{defi}

\noindent
We now make the relation with the notion of origin precise: as shown
in~\cite{DBLP:conf/lics/DartoisFL18}, there is a tight connection between origin
graphs and data words. Here, the encoding is slightly different, as we do not
necessarily ask that the data labelling input position $n$ is equal to $n$.
However, as long as the input data are all pairwise distinct, such encoding
carries to our setting: the output data at position $j$ is equal to $d_\inp^i$,
where $i$ is the (input) origin position. Thus, in the following, we let
$\allDiff$ denote the set of relational data words whose input data are pairwise
distinct:
\begin{equation*}
 \allDiff =
\{w=(\sigma_\inp^1,d_\inp^1)(\sigma_\outp^1,d_\outp^1)\dots\in \relwords \mid
\forall 0\leq i < i', d_\inp^i \neq d_\inp^{i'}\}
\end{equation*}
where, by convention $d_\inp^0 =
\dz$. Then, as we will show, the behaviour of an \nratf over $\allDiff$
determines its origin relation, and hence its behaviour over the entire data
domain.

To a run $\rho = q_0 \xrightarrow{\sigma^1_\inp,\asgn^1,r^1,\sigma^1_\outp} q_1
\xrightarrow{\sigma^2_\inp,\asgn^2,r^2,\sigma^2_\outp} q_2 \dots$, we associate
the origin function $\orig_\rho : j \mapsto \max \{i \leq j \mid r_j \in
\asgn_i\}$, with the convention $\max \varnothing = 0$. In other words,
$\orig_\rho(j)$ is the last input position at which the register output at
position $j$ was assigned, so the corresponding input data is the one which is
output (if the register has never been assigned, it contains $\dz$, which, by
convention, is the data associated with input position $0$).

Now, for an origin function $o : \N \backslash \{0\} \rightarrow \N$ and for a
relational data word $w \in \relwords$, we say $w$ is compatible with the origin
function $o$, denoted $w \models o$, whenever for all $j \geq 1$,
$\dt(\projout(w)[j]) = \dt(\projin(w)[o(j)])$, with the convention $\dt(\projin(w)[0]) = \dz$.

 The following proposition shows that
 actual data values in a word $w$ do not matter with respect to membership in some
 \nratf, only the compatibility with origin functions does:
 \begin{prop}%
   \label{prop:originData}
   Let $w \in \relwords$ and $\rho$ a sequence of transitions of some \nratf. Then,
   \begin{enumerate}
   \item\label{itm:origEqual} If $\rho$ is a run over $w$, then $w \models o_\rho$.
   \item\label{itm:equalOrig} If $\rho$ is a run over $w$ and {\normalfont$w \in \allDiff$}, then for all $o : \N \backslash \{0\} \rightarrow \N$, $w \models o \Leftrightarrow o = o_\rho$.
   \item\label{itm:origRun} If $w$ and $\rho$ have the same finite labels and if $w \models o_\rho$, then $\rho$ is a run over $w$.
   \end{enumerate}
 \end{prop}
 \begin{proof}
   (\ref{itm:origEqual}) and (\ref{itm:origRun}) follow from the semantics of \nratf, which do not conduct any test on the input data. The $\Leftarrow$ direction of (\ref{itm:equalOrig}) is exactly (\ref{itm:origEqual}). Now, assume $w \in \allDiff$ admits $\rho$ as a run, and let $o$ such that $w \models o$. Then, let $j \geq 1$ be such that $\dt(\projout(w)[j]) = \dt(\projin(w)[o(j)])$. By (\ref{itm:origEqual}) we know that $\dt(\projout(w)[j]) = \dt(\projin(w)[\orig_\rho(j)])$, so $\dt(\projin(w)[o(j)]) = \dt(\projin(w)[\orig_\rho(j)])$. Since $w \in \allDiff$, this implies $o(j) = o_\rho(j)$, so, overall, $o = \orig_\rho$.
 \end{proof}
 It is not clear whether $W_{S,k}$ is regular for \nratf specifications, but we
 show that it suffices to consider another set denoted $W_{S,k}^\tf$ which is
 easier to analyse (and can be proven regular), which describes the behaviour of
 $S$ over input with pairwise distinct data. Indeed, as expressed by the above
 proposition, \nratf cannot conduct tests on input data, and their behaviour
 only depends on the input labels. Thus, it suffices to study runs on input
 words whose data are all distinct; such choice ensures that two equal input
 data will not ease the task of the implementation. Otherwise, it could be that
 on reading a data word, two registers $r_1$ and $r_2$ are equal, and then the
 implementation can simultaneously take transitions labelled with
 $\textsf{out}(r_1)$ and $\textsf{out}(r_2)$. An interesting side-product of
 this approach is that it implies that we can restrict to test-free
 implementations. A test-free transducer is a transducer whose transitions do
 not depend on tests over input data, i.e., for all transitions $t = q
 \myxrightarrow{\sigma_\inp, \tst \mid \asgn, \sigma_\outp, r} q' \in \delta$,
 we have $\tst = \top$.
\begin{prop}%
  \label{prop:realTF}
  Let $S$ be a \nratf specification, and $A_\inp^\varnothing = \inalph \times
  \{\varnothing\}$. The following are equivalent:
  \begin{enumerate}
  \item\label{itm:SReal} $S$ is realisable
  \item\label{itm:WSkEmptyReal}
    $W_{S,k}^{\tf} = \{\overline{a} \in {(A_\inp^\varnothing A_\outp^k)}^\omega
    \mid \Comp(\overline{a}) \cap S \cap {\normalfont\allDiff} \neq \varnothing\}$
    is realisable by a (register-free) transducer with input alphabet $A_\inp^\varnothing$
  \item\label{itm:SRealTF} $S$ is realisable by a test-free transducer
  \end{enumerate}
\end{prop}

\begin{proof}
  $(\ref{itm:SRealTF}) \Rightarrow (\ref{itm:SReal})$ is trivial.

  $(\ref{itm:SReal}) \Rightarrow (\ref{itm:WSkEmptyReal})$: If $S$ is realisable,
  then, by Theorem~\ref{thm:infinite2finite}, $W_{S,k}$ is
  realisable by some transducer $I$. Now, since transducers are closed under
  regular domain restriction, $W_{S,k}^\varnothing = W_{S,k} \cap
  {(A_\inp^\varnothing A_\outp^k)}^\omega$
  is realisable by $I$ restricted to the input alphabet $A_\inp^\varnothing$;
  more precisely, by the transducer $I'$ with the same set of states as $I$ and
  transition function $\delta' = \delta \cap \left(Q_I \times \inalph \times
    \{\varnothing\} \rightarrow \Asgn_{R_k} \times \outalph \times R_k \times
    Q_I \right)$. Moreover, $W_{S,k}^\varnothing \subseteq W_{S,k}^{\tf}$.
  Indeed, let $\overline{a} \in W_{S,k}^\varnothing$. Then, $\Comp(\overline{a})
  \subseteq S$. It is easy to build by induction a data word $w \in
  \Comp(\overline{a}) \cap \allDiff$, so $\Comp(\overline{a}) \cap S \cap
  \allDiff \neq \varnothing$. Thus, $W_{S,k}^\tf$ is realisable by any
  transducer realising $W_{S,k}^\varnothing$.

  $(\ref{itm:WSkEmptyReal}) \Rightarrow (\ref{itm:SRealTF})$: Now, assume
  $W_{S,k}^\tf$ is realisable by some transducer $I$. We show that $I$, when
  ignoring the $\varnothing$ input tests, is actually an implementation of $S$.
  Thus, let $I'$ be the same transducer as $I$ except that all input tests
  $\varnothing$ have been replaced with $\top$. Formally, $q
  \myxrightarrow[I']{\sigma_\inp, \top \mid \asgn, \sigma_\outp, r} q'$ iff $q
  \myxrightarrow[I]{\sigma_\inp, \varnothing \mid \asgn, \sigma_\outp, r} q'$
  Note that $I'$, interpreted as a register transducer, is test-free. Let $w \in
  \datawords$, and $\overline{a}_\inp = \lab(w) \times \varnothing^\omega$ be
  the input action in $A_\inp^\varnothing$ with same finite labels as $w$. Let
  $\overline{a} = I(\overline{a}_\inp)$, and let $w' \in \Comp(\overline{a})
  \cap S \cap \allDiff$ (such $w'$ exists because, as above,
  $\Comp(\overline{a}) \cap \allDiff \neq \varnothing$). Then, since
  $\lab(w) = \lab(w')$, they admit the same run $\rho^I$ in $I$, so $w, w'
  \models o_{\rho^I}$. Now, $w' \in S$, so it admits an accepting run $\rho^S$
  in $S$, which implies $w' \models o_{\rho^S}$. Moreover, $w' \in \allDiff$ so,
  by Proposition~\ref{prop:originData} (\ref{itm:equalOrig}), we get $o_{\rho^I} =
  o_{\rho^S}$. Therefore, $w \models o_{\rho^S}$, so, by
  Proposition~\ref{prop:originData} (\ref{itm:origRun}), $w$
  admits $\rho^S$ as a run, i.e.\ $w \in S$. Overall, $L(I) \subseteq S$, meaning
  that $I$ is a (test-free) implementation of $S$.
\end{proof}

Finally, $W_{S,k}^{\tf} = \{\overline{a} \in {(A_\inp^\varnothing
A_\outp^k)}^\omega \mid \Comp(\overline{a}) \cap S \cap \allDiff \neq
\varnothing\}$ is regular. Indeed, $W_{S,k}^\tf = \{\overline{a} \in
{(A_\inp^\varnothing A_\outp^k)}^\omega \mid \Comp(\overline{a}) \cap
S^\varnothing \neq \varnothing\}$, where $S^\varnothing$ is the same automaton
as $S$ except that all input transitions $q \xrightarrow{\sigma_\inp,
  \top, \asgn} q'$ have been replaced with $q \xrightarrow{\sigma_\inp,
  \bigwedge_{r\in R_k} r^{\neq}, \asgn} q'$, because, for all $\overline{a} \in
{(A_\inp^\varnothing A_\outp^k)}^\omega$, $\Comp(\overline{a}) \cap S \cap
\allDiff \neq \varnothing \Leftrightarrow \Comp(\overline{a}) \cap S^\varnothing
\neq \varnothing$ (the $\Rightarrow$ direction is trivial, and the $\Leftarrow$
stems from the fact that an $\allDiff$ input only takes $\tst = \varnothing$
transitions).

  Then, $L_{S,k}^\tf = \{w \otimes \overline{a} \in \relwords \otimes {(A_\inp^\varnothing A_\outp^k)}^\omega \mid w \in \textsf{Comp}(\overline{a}) \cap S^\varnothing\}$ is $\nra$-definable. Indeed, $S$ is $\nratf$-definable, so $S^\varnothing$ is $\nra$-definable, and by Lemma~\ref{lem:compreg}, $L_k = \{w \otimes \overline{a} \mid w \in \textsf{Comp}(\overline{a})\}$ is $\nra$-definable, so their product recognises $L_{S,k}^\tf$.
  Finally, $W_{S,k}^\tf = \lab(L_{S,k}^\tf)$, and the projection of a $\nra$ over some finite alphabet is regular~\cite{Kaminski:1994:FA:194527.194534}.

Overall, by Theorem~\ref{thm:infinite2finite},
we finally get (the complexity analysis is the same as for \ura):
\begin{thm}\label{thm:boundedNRAtf}
  For all $k\geq 0$, {\normalfont$\real(\nratf, \rt[k])$} is decidable and in {\normalfont\textsc{2ExpTime}}.
\end{thm}

\section{Synthesis and Uniformisation}%
\label{sec:domain}

In this section, we discuss the connection between synthesis and
uniformisation of relations, which is a more general problem: as pointed out in Section~\ref{sec:synthesisProblem},
if $S$ is realisable by a register transducer, then, in particular, it has a
total domain, i.e.\ $\projin(S) = \datawords(\inalph, \data)$, otherwise it
cannot be that $L(T) \subseteq S$ for $T$ a register transducer, since by
definition of transducers $\projin(T) = \datawords(\inalph, \data)$. However, when defining a
specification, the user might be interested only in a subset of behaviours (for
instance, s/he knows that all input data will be pairwise distinct). In the
finite alphabet setting, since the formalisms used to express specifications are
closed under complement (whether it is LTL or $\omega$-automata), it is actually
not a restriction to assume that the input domain of the specification is total:
it suffices to
complete the specification by allowing any behaviour on the input not
considered. However, since register automata are not closed under complement,
such approach is not possible here. Thus, it is relevant to generalise the
realisability problem to the case where the domain of the specification is not
total. This can be done by equipping register transducers with an acceptance
condition. It is also necessary to adapt the notion of realisability; otherwise,
any transducer accepting no words realises any specification. (since it is
always the case that $\varnothing \subseteq S$). A natural way is to consider
synthesis as a uniformisation problem~\cite{DBLP:conf/icalp/FiliotJLW16}. An (implementation)
function $f : \In \rightarrow \Out$ is said to \emph{uniformise} a (specification)
relation $R \subseteq \In \times \Out$ whenever:
\begin{enumerate}
  \item\label{itm:domain} $\dom(f) = \dom(R)$ and
  \item\label{itm:subrel} for all $i \in \dom(f), (i, f(i)) \in R$
  \end{enumerate}
  Note that constraint~\ref{itm:domain} is the main difference with the
  notion of realisability.

  In the context of reactive synthesis, where $f = f_I$ is defined from an
implementation $I$ and $R$ is given as a language of relational words, it can be
rephrased as
\begin{enumerate}
  \item $\projin(L(I)) = \projin(R)$ and
  \item for all $w_\inp \in \projin(L(I)), \interl{w_\inp}{f_I(w_\inp)} \in R$
  \end{enumerate}
  Note that such definition coincides with the one of realisability of
  Section~\ref{sec:synthesisProblem} when the class of implementations has
  total domain, because then it is equivalent to asking $L(I) \subseteq
  R$. In the following, we denote by $\unif(\mathcal{S},\mathcal{I})$ the
  uniformisation problem from specifications in $\mathcal{S}$ to implementations
  in $\mathcal{I}$.
Unfortunately, this setting is actually much harder, as shown by the next
two theorems:
\begin{thm}%
  \label{thm:UndecDomain}
  Given $S$ a specification represented by a \dra, checking whether $\projin(S) =
  \datawords(\inalph, \data)$ is undecidable.
\end{thm}
\begin{proof}
  We reduce from the universality problem of \nra, which is
  undecidable~\cite{Neven:2004:FSM:1013560.1013562}. Let $A = (\Sigma, \data, Q,
  q_0, \delta, R, c)$ be an \nra. We encode $L(A)$ as the domain of some \dra
  specification: the input transitions are the same as the transitions of the
  original automaton, but when there is some nondeterminism, its resolution is
  postponed to the corresponding output transition, whose finite label
  corresponds to the chosen transition. In the vocabulary of games, the input
  player chooses the finite input label and the equality relation of the input
  data to the registers of $A$, and the output player resolves the
  nondeterminism. Thus, we construct a \dra $D$ accepting $R(D) = \{((\sigma_1, d_1)
  (\sigma_2, d_2) \dots, (t_1, d_1) (t_2, d_2) \dots) \mid t_1 t_2 \dots$
    is a run of $A$ over $(\sigma_1, d_1) (\sigma_2, d_2) \dots\}$.

    Thus, define $D = (\Sigma \uplus \delta, \data, Q \uplus Q \times (\Sigma
    \times \Tst_R), q_0, \delta', R \uplus \{r_0\}, c')$, where $\delta'$ is
    defined as follows: for all $q \in Q$, $\sigma \in \Sigma$ and $\tst \in
    \Tst_R$, we define the input transition $q \myxrightarrow[D]{\sigma, \tst,
      \{r_0\}} (q, (\sigma, \tst))$. Then, for all $t = q
    \myxrightarrow[A]{\sigma, \tst, \asgn} q' \in \delta$, we define the output
    transition $(q, (\sigma, \tst)) \myxrightarrow[D]{t, \tst
      \wedge r_0^=,
      \asgn} q'$. Then, let $c' : q \mapsto c(q)$ and $(q,\bullet) \mapsto
    c(q)$. Such automaton is indeed deterministic, and it recognises the
    relation $R(D) = \{((\sigma_1, d_1) (\sigma_2, d_2) \dots, (t_1, d_1) (t_2,
    d_2) \dots) \mid t_1 t_2 \dots$ is a run of $A$ over $(\sigma_1, d_1)
    (\sigma_2, d_2) \dots\}$. Then, $\projin(R(D))$ is universal iff $L(A)$ is
    universal.
  \end{proof}
  Such result extends to \nra and \ura, whose \dra are a special case. Note that
  the unbounded realisability problem for $\dra$ is not reducible to deciding
  whether the domain is total: if the specification $S$ is not realisable, it is not
  possible to determine whether it is because the domain of $S$ is not total or
  because $S$ is not realisable by a sequential machine (e.g.\ $S$ asks to output
  right away a data that will only be input in the future).

Then, while the uniformisation setting obviously preserves the undecidability
results from the synthesis setting, the above result allows to show that the
somehow more general uniformisation problem is undecidable. For instance, we can prove:
\begin{thm}%
  \label{thm:boundedURAFinDom}
  For all $k \geq 1$, {\normalfont$\unif(\ura,\rt[k])$} is undecidable.
\end{thm}
\begin{proof}
Consider some unrealisable \ura specification $S_u$ and the following
specification $S$ mapping $w_1\# w_2$ to $w_1 \# w'_2$ such that $(w_2,w'_2)\in
S_u$, defined only when $w_1$ is a finite data word accepted by some \ura $A$.
Clearly, $S$ is \ura-definable and realisable iff its domain is empty, i.e.
$L(A) = \varnothing$. However, emptiness of \ura is an undecidable problem.
\end{proof}

If the domain of the specification is \dra-recognisable, it is possible to
reduce the uniformisation problem to realisability, by allowing any behaviour on
the complement of the domain (which is then \dra-recognisable). However, such
property is undecidable as a direct corollary of Theorem~\ref{thm:UndecDomain}.

\section*{Conclusion}%
\label{sec:conclusion}

In this paper, we have given a picture of the decidability landscape of the
synthesis of register transducers from register automata specifications. We
studied the parity acceptance condition because of its generality, but our
results allow to reduce the synthesis problem for register automata
specifications to the one for finite automata while preserving the acceptance
condition. We have also introduced and studied test-free \nra, which do not have
the ability to test their input, but still have the power of duplicating,
removing or copying the input data to form the output. We have shown that they
allow to recover decidability in the presence of non-determinism, in the bounded
synthesis case. We leave open the unbounded case, which we conjecture to be
decidable. As future work, we want to study synthesis problems for register
automata which are able to test additional properties over the data. In
particular, allowing to compare data for an order over
$\data$~\cite{DBLP:conf/amw/BenediktLP10,DBLP:journals/mscs/FigueiraHL16} looks
promising. Note that most other natural predicates immediately yield
undecidability, e.g.\ adding +1. Another direction is to study specifications
given by logical formulae, for decidable data words logics such as two-variable
fragments of
FO~\cite{DBLP:conf/lics/BojanczykMSSD06,DBLP:journals/corr/abs-1110-1439,DBLP:conf/lics/DartoisFL18}.
Such problem is however much more challenging, as there do not exist good
correspondence between logic and automata in the realm of data words, except in
very restricted settings~\cite{DBLP:conf/csl/BenediktLP10}.

\section*{Acknowledgments}
\noindent The authors would like to thank Ayrat Khalimov for his remarks and
suggestions, which helped improve the quality of the paper. They also thank the
anonymous reviewers, who took the time to read the paper in detail and
subsequently suggested important clarifications as well as simplifications in
the proofs.



\bibliographystyle{alpha}
\bibliography{../Bibliography}

\end{document}